\documentclass[sigconf,nonacm]{acmart}

\setcopyright{none}        
\acmConference[Preprint]   
              {Under Review}{2025}{}

\acmDOI{}   
\acmISBN{}   
\usepackage{algorithm}
\usepackage{algorithmic}
\usepackage{booktabs}
\usepackage{newfloat}
\usepackage{listings}
\DeclareCaptionStyle{ruled}{labelfont=normalfont,labelsep=colon,strut=off} 
\floatstyle{ruled}
\newfloat{listing}{tb}{lst}{}
\floatname{listing}{Listing}
\usepackage{amsthm}
\usepackage{amsmath}

\usepackage{amssymb}
\usepackage{booktabs}
\usepackage{placeins}

\theoremstyle{definition}
\newtheorem{definition}{Definition}[section]

\theoremstyle{plain}
\newtheorem{theorem}{Theorem}[section]
\newtheorem{lemma}{Lemma}[section]

\title{ Verifiable Exponential Mechanism for  Median Estimation}

\author{Hyukjun Kwon}
\affiliation{
  \institution{Seoul National University}
  \city{Seoul}
  \country{Republic of Korea}
}
\email{todd4@snu.ac.kr}

\author{Chenglin Fan}
\affiliation{
  \institution{Seoul National University}
  \city{Seoul}
  \country{Republic of Korea}
}
\email{fanchenglin@snu.ac.kr}

\usepackage{bibentry}

\begin{document}

\begin{abstract}
  Differential Privacy (DP) is a rigorous privacy standard widely adopted in data analysis and machine learning. However, its guarantees rely on correctly introducing randomized noise—an assumption that may not hold if the implementation is faulty or manipulated by an untrusted analyst. To address this concern, we propose the first verifiable implementation of the exponential mechanism using zk-SNARKs. As a concrete application, we present the first verifiable differentially private (DP) median estimation scheme, which leverages this construction to ensure both privacy and verifiability. Our method encodes the exponential mechanism and a utility function for the median into an arithmetic circuit, employing a scaled inverse CDF technique for sampling. This design enables cryptographic verification that the reported output adheres to the intended DP mechanism, ensuring both privacy and integrity without revealing sensitive data.
\end{abstract}


\maketitle

\section{Introduction}

Modern digital services collect extensive personal data—from GPS and health records to shopping and transit history—to provide tailored conveniences like personalized healthcare and product recommendations. However, such benefits often come at the cost of privacy~\citep{Zang2011Anonymization}. Anonymization alone is insufficient: in the 2006 Netflix Prize, \citep{Narayanan_Netflix} re-identified users from an anonymized dataset using a few known ratings and timestamps.

To address these risks, Differential Privacy (DP)~\citep{DP2006_Dwork, DworkRoth2014} provides formal privacy guarantees by adding carefully calibrated noise, ensuring that the output reveals little about any individual. DP has been widely adopted across domains such as optimization~\citep{gupta2010differentially}, reinforcement learning~\citep{vietri2020private}, and deep learning~\citep{abadi2016deep}. However, because DP introduces randomness, it becomes difficult to verify whether the reported output was computed honestly. A malicious party could distort outputs and attribute deviations to DP noise—creating a gap between privacy and integrity.

Local Differential Privacy (LDP)~\citep{Kasiviswanathan2011WhatCanWeLearnPrivately, Evfimievski2003} avoids a trusted aggregator by having users perturb their data locally. While LDP ensures strong individual privacy, it often suffers from lower utility. Recent work has explored \emph{verifiable} LDP using cryptographic proofs~\citep{Kato0Y21, GarridoSB22, Bontekoe2024}, especially for categorical data via randomized response~\citep{Randomized_Response_Warner1965, Crypto_RR_AmbainisJL03}, but such protocols remain computationally expensive and hard to scale.

In contrast, verifiability in the central DP model has so far focused on numerical queries using mechanisms like Laplace~\citep{Biswas2023, WeiCYLP25}, with no general verifiable protocol for categorical queries requiring the exponential mechanism.

This work closes that gap by introducing the first verifiable exponential mechanism in the central DP setting, enabling efficient, sound verification of privacy-preserving categorical data release.

\textbf{Our Contributions.} We extend verifiable differential privacy in the central model to categorical queries by introducing the first verifiable Exponential Mechanism (EM) for median estimation using zk-SNARKs. Our key contributions are:

\begin{enumerate}
    \item \textbf{Logic and Circuit Design for Verifiable EM:} We design an arithmetic circuit for the exponential mechanism tailored to median estimation, using modular subcircuits that implement a scaled inverse CDF sampling algorithm based on the utility function.
    
    \item \textbf{Public Verifiability:} We integrate zk-SNARKs with the Poseidon hash to enable any Verifier to efficiently verify correct median sampling without trusting the Prover.
    
    \item \textbf{Non-Interactive Protocol:} Our scheme is fully non-intera\-ctive, relying on a public bulletin board and Data Provider-supplied randomness, eliminating the need for online interaction.
    
    \item \textbf{Formal Guarantees:} We prove the scheme satisfies security, differential privacy, and strong utility guarantees.
    
    \item \textbf{Efficient Implementation:} Using Groth16 zk-SNARKs and Poseidon, we demonstrate practicality: for datasets of up to 7,000 entries, proofs are generated in minutes and verified in under a second.
\end{enumerate}


    
    
    
    

\subsection{Related Works}

\textbf{Differential Privacy (DP)} is a rigorous framework for ensuring individual privacy in statistical analysis~\citep{DP2006_Dwork,dwork2006our}. The \emph{Laplace Mechanism}~\citep{DP2006_Dwork} adds calibrated Laplace noise to numeric query results. For categorical queries, the \emph{Exponential Mechanism}~\citep{McSherry2007} samples outputs according to a utility-based distribution, enabling DP in non-numeric domains.

\textbf{Verifiable Differential Privacy (VDP)} enhances DP by allowing third parties to verify compliance. \citep{NarayanFPH15} introduced this concept but lacked support for categorical queries. \citep{Biswas2023} proposed using zk-SNARKs with binomial noise for verifiable counting queries, while \citep{WeiCYLP25} extended this to general numerical queries via the Laplace mechanism. However, support for categorical outputs and the Exponential Mechanism remains an open challenge.

\textbf{Verifiable Local Differential Privacy (VLDP)} ensures each user perturbs their data locally~\citep{Evfimievski2003,Kasiviswanathan2011WhatCanWeLearnPrivately}, often using Randomized Response (RR)~\citep{Randomized_Response_Warner1965,Crypto_RR_AmbainisJL03}. Verifying local perturbations introduces high overhead, as each user must generate a zero-knowledge proof. Recent advances include verifiable binary RR~\citep{Kato0Y21}, polling schemes with zk-SNARKs~\citep{GarridoSB22}, and verifiable \(k\)-ary RR in the shuffle model~\citep{Bontekoe2024}.
Other related research, though not focused on differential privacy, includes Zero-Knowledge Proofs for Large Language Models \citep{SunL024} and Deep Learning~\citep{SunBLZ25}.

\subsection{Technique Overview of Our Approach}

This paper addresses the challenge of verifying that a differentially private  median estimation is both \emph{accurate} and \emph{private}. To achieve this, the authors combine the exponential mechanism (used for DP in categorical queries) with zk-SNARKs, a form of zero-knowledge proof. The probabilistic sampling of the exponential mechanism is encoded into a static arithmetic circuit, which enforces correct execution by design. Precomputed lookup tables approximate exponential values to support sampling within finite fields, while inverse CDF sampling ensures that the output adheres to the desired DP distribution. By generating a succinct proof (via zk-SNARKs) that the circuit was executed faithfully, any third party can verify that the reported median satisfies differential privacy guarantees—without accessing the raw data or compromising privacy.

\section{Preliminaries}

In this section, we introduce the notation and background concepts used throughout this paper. We begin by fixing our basic notation. We then give a concise overview on and background knowledge on security, cryptography and privacy. Finally, we provide a simple description of inverse CDF sampling, which serves as our primary mechanism for generating differentially private outputs.

\subsection{Notations}

Let $\mathcal{X}$ be the set of possible input categories. $\mathcal{DB}\in\mathbb{N}^{|\mathcal{X}|}$ denotes a database of $m$ records. We define "Adjacent Databases" as two databases that differ in exactly one record and denote $\mathcal{DB}\sim \mathcal{DB}'$. $\mathcal{R}$ denotes the range of possible outputs of our mechanism or query. $\mathcal{M}$ denotes a mechanism or query that maps a database to an element of $\mathcal{R}$. We denote the utility function of the exponential mechanism as $u:\mathbb{N}^{|\mathcal{X}|} \times\mathcal{R}\rightarrow\mathbb{R}$. We denote $\mathbb{F}_p$ as a finite field of order prime $p$, where our scheme is defined on. For an integer \(k>0\), \([k]\) denotes the set of integers from \(0\) to \(k-1\). \(||\) denotes concatenation. For a set \(X\), \(\mathcal{U}(X)\) denote the uniform distribution over \(X\). We denote sampling \(x\) from a distribution \(Y\) as \(x\overset{\$}{\leftarrow}Y\).

\subsection{Background on Security and Cryptography}

\subsubsection{Negligiblity and Indistinguishability}

\begin{definition}[Negligibility]
    A function \(\mathsf{negl}:\mathbb{R}\rightarrow\mathbb{R}\) is negligible if for all polynomial \(q(n)\), there exists \(N\in\mathbb{N}\) such that
    \[
        \forall n>N, \mathsf{negl}(n)< \frac{1}{q(n)}
    \]
\end{definition}

\begin{definition}[Computational Indistinguishability]
    Fix a security parameter \(\lambda\in\mathbb{N}\), and probability distributions \(\{X_n\}_{n\in\mathbb{N}},\{Y_n\}_{n\in\mathbb{N}}\) over \(\{0,1\}^{poly(\lambda)}\). We say \(\{X_n\}_{n\in\mathbb{N}}\) and \(\{Y_n\}_{n\in\mathbb{N}}\) are computationally indistinguishable if for all PPT (probabilistic polynomial time) distinguisher \(D\), there exists a negligible function \(\mathsf{negl}(\lambda)\) such that
    \[
        \left|\Pr[D(X_\lambda)=1]-\Pr[D(Y_\lambda)=1]\right| \le \mathsf{negl}(\lambda)
    \]
    We denote the computational indistinguiability as the following notation.
    \[
        \{X_n\}_{n\in\mathbb{N}} \overset{c}{\approx} \{Y_n\}_{n\in\mathbb{N}}
    \]
\end{definition}

\subsubsection{Zero-Knowledge Proofs}

Informally, Zero-Knowledge Proofs (ZKP) \citep{Goldwasser1985} is a technique to prove that a statement is true without leaking any other information about the statement other than the fact that the statement is true. 

We denote the prover as \(\mathcal{P}\) and the verifier as \(\mathcal{V}\). Formally, a zero-knowledge proof system $(\mathcal{P},\mathcal{V})$ must satisfy the following properties.
\begin{definition}[Zero-Knowledge Proofs] \label{def:security definition of ZKP}
Let \(R\subseteq X\times W\) be a relation. We say that a statement \(x\) is true if \((x,w) \in R\) for some witness \(w\in W\). Zero-knowledge proof is a proof system \((\mathcal{P}(x,w), \mathcal{V}(x))\) that satisfies the following properties.
    \begin{itemize}
        \item Completeness: If \((x,w) \in R\), an honest verifier accepts the proof with high probability.
        \item Soundness: If a prover does not know \(w\) such that \((x,w)\in R\), an honest verifier rejects with high probability.
        \item Zero-Knowledge: For every verifier \(\mathcal{V}^*\), there exists a PPT algorithm \(\mathsf{Sim}_{\mathcal{V}^*}\) on input \(x\) such that 
        \[\mathsf{Sim}_{\mathcal{V}^*}(x)\approx \mathsf{View}_{\mathcal{V}^*}[\mathcal{P}(x,w)\leftrightarrow\mathcal{V}^*(x)]\] 
    \end{itemize}
\end{definition}

By ``honest'', we mean that the party faithfully follows the ZKP protocol. \(\mathsf{View}_{\mathcal{V}^*}\) denotes the ``view'' of \(\mathcal{V}^*\), which includes the transcript between \(\mathcal{V}^*\) and \(\mathcal{P}\), secret values of \(\mathcal{V}^*\), and internal randomness of \(\mathcal{V}^*\). If the distribution of the simulated view and the real view is identical, we say it is ``perfect'' zero-knowledge. If the two distributions are computationally indistinguishable, we say it is ``computational'' zero-knowledge.

\subsubsection{zk-SNARKs}

zk-SNARKs (Zero-Knowledge Succinct Non-Interactive Arguments of Knowledge) are a class of non-interactive zero knowledge proofs known for their succinctness and efficiency. The first practical system, Pinocchio~\citep{Pinocchio_ParnoBH13}, was followed by several constructions; we adopt Groth16~\citep{Groth16} for its compact proofs and fast verification.

To generate a zk-SNARK proof, the prover encodes the target statement as an arithmetic circuit or constraint system (e.g., R1CS or QAP) over a finite field \(\mathbb{F}_p\), then compiles it with a ``witness'' into algebraic objects such as polynomials. Verification involves checking that a polynomial identity holds at a randomly chosen point, ensuring correctness without revealing the witness. Formally, we can express a zk-SNARK scheme as follows.



\begin{itemize}
    \item \((\mathsf{pp}, \mathsf{pk}, \mathsf{vk}) \leftarrow\mathsf{zkSNARK.Setup}(1^\lambda,\mathcal{C})\)\\ Receiving the security parameter \(\lambda\) and the circuit \(\mathcal{C}\), the setup function generates public parameter \(\mathsf{pp}\), proving key \(\mathsf{pk}\), and verification key \(\mathsf{vk}\). In this step, information about the circuit \(\mathcal{C}\) is put into \(\mathsf{pp}, \mathsf{pk}\) and \(\mathsf{vk}\).
    \item \((Output, \pi)\leftarrow\mathsf{zkSNARK.Prove}(\mathsf{pp},\mathsf{pk}, x, w)\)\\ Receiving public parameter \(\mathsf{pp}\), proving key \(\mathsf{pk}\), public input \(x\), and witness \(w\), the proof generation function generates a proof \(\pi\) that the circuit \(\mathcal{C}\) is satisfiable with the witness \(w\). \(Output\) is the output of the circuit \(\mathcal{C}\) when the assignment to each wire is set according to \(w\).
    \item \(b\leftarrow \mathsf{zkSNARK.Verify}(\mathsf{pp}, \mathsf{vk}, x, y, \pi)\)\\ Receiving public parameter \(\mathsf{pp}\), verification key \(\mathsf{vk}\), public input \(x\), output \(y\), and the proof \(\pi\), the verification function outputs a bit that indicates either \(accept\) or \(reject\).
\end{itemize}
Public parameter \(\mathsf{pp}\) is always required for proof generation and verification. Therefore, \(\mathsf{pp}\) is often omitted for brevity.

Schwartz–Zippel lemma underpins the soundness guarantee by bounding the probability that an incorrect witness satisfies all evaluations in the verification is negligible.

\begin{lemma}[Schwartz–Zippel Lemma] 
Let $P(x_1, x_2, ..., x_n)$ be a non-zero polynomial of degree $d$ over a field \(\mathbb{F}\).

Let \( S \subseteq \mathbb{F} \) be a finite subset. Then:
\[
\Pr_{(r_1, \ldots, r_n) \in S^n} \left[ P(r_1, \ldots, r_n) = 0 \right] \leq \frac{d}{|S|}
\]
where $r_1, ..., r_n$ are chosen independently and uniformly at random from $S$.
\end{lemma}

Since this is an extremely summarized explanation, one is encouraged to refer to additional resources \citep{Pinocchio_ParnoBH13, Groth16, Plonk_GabizonWC19} to understand the underlying mechanism of zk-SNARKs for achieving zero-knowledge and succinct proofs. Also, it is worth noting that all existing zk-SNARKs are ``computational'' zero-knowledge since they rely on cryptographic assumptions such as pairing and Diffie-Hellman assumption to achieve succinct proof and fast verification time.
  
\subsubsection{Cryptographic Hash Function}

\begin{definition}[Cryptographic Hash Function]
    Cryptographic Hash Function is a hash function $H:\mathcal{X}\rightarrow \mathcal{Y}$ that satisfies the following properties.
    \begin{itemize}
        \item Collision Resistance: It is infeasible to find  $x\neq x'\in \mathcal{X}$, such that $H(x)=H(x')$.
        \item Preimage Resistance: Given $y \in \mathcal{Y}$, it is infeasible to find $x\in\mathcal{X}$ such that $H(x)=y$.
        \item Second-Preimage Resistance: Given $x\in \mathcal{X}$, it is infeasible to find $x'\in\mathcal{X}$ such that $H(x)=H(x')$.
    \end{itemize}
\end{definition}

We only use cryptographic hash as a hash function.

\subsubsection{Random Oracle Model}

\begin{definition}[Random Oracle Model]
    For each new input $x\in\mathcal{X}$, hash function $H:\mathcal{X}\rightarrow\mathcal{Y}$ outputs a uniformly random output from $\mathcal{Y}$. For any repeated input $x'\in \mathcal{X}$, the hash outputs the previously sampled value.
\end{definition}
Later we prove the security of our scheme in Random Oracle Model.
The assumption of Cryptographic Hash and the Random Oracle Model is crucial for guaranteeing the soundness and zero-knowledge of our scheme.

\subsubsection{Commitment Scheme}

To guarantee the correctness of the output, we need a mechanism for binding so that no party can change the data for the analysis during the protocol. Also, information about the data must not be leaked in order to achieve zero-knowledge. Such requirements are achieved using commitment schemes. 

\begin{definition} [Commitment Scheme]
    Commitment scheme is a triple of algorithms \((\mathsf{Gen}, \mathsf{Com}, \mathsf{Open})\) over sets \((M,R,C)\).
    \begin{itemize}
        \item \(\mathsf{pp}\leftarrow\mathsf{Gen}(1^\lambda)\) is the setup algorithm that generates public parameters \(\mathsf{pp}\).
        \item \(\mathsf{Com}:M\times R\rightarrow C\) is the commit algorithm that commits to a message in the message space \(M\) using a randomness in the randomness space \(R\) into a commitment in the commitment space \(C\).
        \item \(\mathsf{Open}: C\times R \times M\rightarrow \{accept, reject\}\) is the opening algorithm that outputs \(\mathsf{Open}(c\in C,r\in R, m\in M)=accept\) if \(\mathsf{Com}(m, r)=c\) and outputs \(reject\) otherwise.
    \end{itemize}
\end{definition}

A commitment scheme satisfies the following properties.

\begin{enumerate}
    \item Binding: It is computationally infeasible to find \(c\in C, m_1\neq m_2\in M, \text{ and } r_1,r_2\in R\) such that \[\mathsf{Open}(c, r_1, m_1)=\mathsf{Open}(c,r_2,m_2)=accept\]
    \item Hiding: A commitment \(c=\mathsf{Com}(m,r)\) should not reveal any information about the committed message \(m\).
\end{enumerate}

\subsubsection{Hash-Based Commitment}

One can implement a commitment scheme using cryptographic hash function \(H\) as follows.
\[
    \mathsf{Com}(x,r):=H(x||r)
\]
The binding property holds due to the collision resistance property of cryptographic hash function. The hiding property holds under the Random Oracle Model.

\subsection{Background on Privacy}

\subsubsection{Differential Privacy}
Differential Privacy (in short, DP) \citep{dwork2006our} guarantees that the distribution of the outputs of a mechanism regarding any two adjacent databases cannot differ more than a certain amount. First introduced by Dwork et al., DP has become a standard for data privacy preservation. For standard analysis of DP, parameters $\varepsilon$ and $\delta$ are used. 

\begin{definition}[$(\varepsilon, \delta)$-Differential Privacy]
    A mechanism $\mathcal{M}$ satisfies $(\varepsilon,\delta)$-DP if for all $\mathcal{DB}\sim\mathcal{DB'}$ and all $S\subseteq\mathcal{R}$:
    \begin{equation*}
        \Pr[\mathcal{M}(\mathcal{DB})\in S] \le e^\varepsilon\cdot \Pr[\mathcal{M}(\mathcal{DB}')\in S]+\delta
    \end{equation*}
\end{definition}

\begin{definition}[IND-Computational Differential Privacy]
    \leavevmode
    Fix a security parameter \(\lambda\in\mathbb{N}\) and let \(\mathcal{M}=\{\mathcal{M}_\lambda\}_{\lambda\in\mathbb{N}}\) be a family of randomized algorithms. We say that \(\mathcal{M}\) satisfies computational \((\epsilon,\delta)\)-DP if for every PPT distinguisher \(D\), for all $\mathcal{DB}\sim\mathcal{DB'}$ and for all $S\subseteq\mathcal{R}$:
    \begin{align*}
        \Pr[D(\mathcal{M}_\lambda(\mathcal{DB}&)\in S)=1] \\
        &\le e^\varepsilon\cdot\Pr[D(\mathcal{M}_\lambda(\mathcal{DB}')\in S)=1]+\delta(\lambda)
    \end{align*}
\end{definition}

In this paper, we only use IND-CDP for computational differential privacy.

Smaller $\varepsilon$ means stronger privacy. If $\delta=0$, we call it Pure Differential Privacy. On the other hand, if $\delta>0$, we call it Approximate Differential Privacy. Also, in Computational DP model, we often say for brevity that \(\mathcal{M}\) is \(\varepsilon\)-DP when \(\delta(\lambda)\) is negligible. 

The \(\ell_1\)-sensitivity of the query determines how much noise is required to guarantee $(\varepsilon,\delta)$-DP.

\begin{definition}[$\ell_1$-Sensitivity]
    The $\ell_1$-sensitivity of a query function $f$ is defined by:
    \begin{equation*}
        \Delta f=\max_{\mathcal{DB}\sim\mathcal{DB}'}\|f(\mathcal{DB})-f(\mathcal{DB}')\|_1
    \end{equation*}
\end{definition}

\subsubsection{Exponential Mechanism}

Exponential Mechanism (EM) \cite{McSherry2007} is a DP mechanism that is specialized in categorical queries. EM rates each element $r$ in the range of the query by a utility function $u(\mathcal{D},r)$, and outputs an element $r'\in \mathcal{R}$ with probability proportional to the exponential of the utility score of $r'$. We present formal definition of the sensitivity of Utility Function and resulting EM.

\begin{definition}[Sensitivity of Utility Function]
    The sensitivity of a utility function $u$ is defined by:
    \begin{equation*}
        \Delta u=\max_{r\in \mathcal{R}}\max_{\mathcal{DB}\sim\mathcal{DB}'}\left|u(\mathcal{DB},r)-u(\mathcal{DB}',r)\right|_1
    \end{equation*}
\end{definition}

\begin{definition}[Exponential Mechanism] \label{def:Exponential Mechanism}
    The Exponential Mechanism $\mathcal{M}$ selects and outputs an element $r\in\mathcal{R}$ with probability as follows:
    \begin{equation*}
        \Pr(\mathcal{M}(\mathcal{DB},u,\mathcal{R})=r')=\frac{\exp\left(\frac{u(\mathcal{DB},r')}{2\Delta u}\right)}{\sum_{r \in \mathcal{R}}\exp\left(\frac{u(\mathcal{DB},r)}{2\Delta u}\right)}
    \end{equation*}
\end{definition}

\begin{theorem} \label{thm:Exponential Mechanism Privacy}
    The Exponential Mechanism of Definition~\ref{def:Exponential Mechanism} satisfies $(\varepsilon,0)$-DP.
\end{theorem}

\subsubsection{Median Query and DP Median Estimation}

\begin{definition}[Median Query]
    \begin{equation*}
        \text{Median}(\mathcal{DB}) = 
        \begin{cases}
            x_{\frac{m+1}{2}} & \text{if } m \text{ is odd} \\
            \frac{1}{2}(x_{\frac{m}{2}}+x_{\frac{m}{2}+1})  & \text{if } m \text{ is even}
        \end{cases}
    \end{equation*}
\end{definition}
In this paper, we use the exponential mechanism for DP median estimation. First, we define the rank function $\text{rank}_\mathcal{DB}(r)$ and the utility function $u(\mathcal{DB},r)$ as follows.
\begin{definition}[Rank and Utility Function]\label{def:util and rank}
    \begin{align*}
        \text{rank}_\mathcal{DB}(r) &= \left| \{x\in \mathcal{DB} \mid x < r\} \right| \\
        u(\mathcal{DB},r) &= -\left|\text{rank}_\mathcal{DB}(r)-\frac{n-1}{2}\right|
    \end{align*}
\end{definition}

\begin{lemma}[Sensitivity of $u(\mathcal{DB},r)$] \label{lem:sensitivity of utility func.}
    The sensitivity of $u(\mathcal{DB},r)$ is 1. In other words, $\Delta u=1$.
\end{lemma}

 Now we describe a DP median estimation mechanism $\mathcal{M}_{median}$ based on exponential mechanism.

\begin{definition}[DP Median Estimation]
    The DP median estimation $\mathcal{M}_{median}(\mathcal{DB},u,\mathcal{R})$ outputs an element $r'\in\mathcal{R}$ with probability proportional to $\exp\left(\frac{\varepsilon\cdot u(\mathcal{DB},r')}{2}\right)$ where the utility function $u$ is defined as Definition~\ref{def:util and rank}.
\end{definition}

\subsection{Inverse CDF Sampling}
Our scheme performs Inverse CDF Sampling from the distribution derived from the exponential mechanism. The definition of Inverse CDF Sampling is as follows. 

\begin{definition}[Inverse CDF Sampling]
    Assume that the CDF of a distribution $D$ is $F_D$. To sample from the distribution $D$, draw $U\sim \text{Uniform}(0,1)$ and output $Y=F_D^{-1}(U)$.
\end{definition}

\section{Security Model}
There are four roles—\textbf{Data Providers} \(\{\mathcal{D}_i\}_{i}\), a \textbf{Data Analyst (Prover)} \(\mathcal{P}\), a \textbf{Verifier} \(\mathcal{V}\), and a \textbf{Public Bulletin Board} \(\mathcal{B}\). All parties other than \(\mathcal{B}\) have private source of randomness. \(\mathcal{B}\) does not need any. 
We assume that all communication between the parties is secure against eavesdropping and forgery.

\textbf{Data Providers}  \(\{\mathcal{D}_i\}_{i}\) are the statistical population. Before protocol execution, each Data Provider \(\mathcal{D}_i\) with private input \(x_i\in\mathcal{R}\) independently samples private randomness $r_i\in\mathbb{F}_p$ and posts a commitment $\mathsf{Com}(x_i, r_i)$ on the public bulletin board. When the protocol begins, each of them sends $(x_i, r_i)$ to the Data Analyst. We assume every Data Provider follows the protocol exactly as specified and does not collaborate with any other party. Modeling collusion, strategic deviation, or dishonest behavior among the Data Providers is outside the scope of this work.

\textbf{Data Analyst} \(\mathcal{P}\) receives the data from the Data Providers and publish the verifiable and differentially private analysis to the dataset. It is the main adversary in our verifiable differential privacy framework. The Analyst may maliciously publish incorrect statistical outputs with counterfeit proofs. Therefore, the central goal of our protocol is to ensure that no Data Analyst can produce a false output or proof that deceives the Verifier into accepting an invalid result.

\textbf{Verifier} \(\mathcal{V}\) checks the proof’s validity and is curious about the underlying private data. To prevent any leakage of information to \(\mathcal{V}\), our protocol satisfies the zero-knowledge property, ensuring the Verifier learns nothing beyond the correctness of the result.

\textbf{The Bulletin Board} \(\mathcal{B}\) is a broadcast channel with an append-only log. Anyone can post a message on the board and the message is then identified with a unique index \(i\). Anyone can read the message from the board by sending the corresponding index \(i\) to \(\mathcal{B}\). We assume the ownership of all messages on the board is publicly visible. The board provides the following guarantees:
\begin{itemize}
    \item Immutability: Once a message is posted on the board, no one can move, alter or erase it.
    \item Global Consistency: Once a message is posted on the board at index \(i\), any reader obtains identical retrieval.
    \item Availability: A single read operation takes constant time.
\end{itemize}
An ideal public bulletin board can be implemented with a permissionless blockchain, which enforces immutability and global consistency via consensus mechanism and cryptographic hashing, while still allowing O(1) access to any indexed entry.

\section{Verifiable Exponential Mechanism}

We present our main technical contribution: the design and implementation of \(\mathsf{VerExp}\), a verifiable exponential mechanism for median estimation. This section includes a high-level overview, and implementation; architecture of the circuit, and core design choices. Detailed description of submodules of the main circuit is provided in the Appendix. 

\subsection{Pipeline}

We assume $m$ and $n$ are positive integers. The overall pipeline is illustrated in Table~\ref{table:pipeline}.

The protocol involves \(m\) Data Providers and a single Data Analyst. The number of Verifiers remains unspecified due to public verifiability. For each \(i\in[m]\), the \(i\)-th Data Provider holds a private integer input \(x_i\in\mathcal{R}\) and independent randomness value \(r_i\in\mathbb{F}_p\). The goal is to generate a differentially private median estimation of \(\{x_i\}_{i\in[m]}\). Note that the range \(\mathcal{R}\) is displayed as \(\{\mathsf{range}_i\}_{i\in[n]}\) for better understanding.

Verifiable exponential mechanism for median estimation \((\mathsf{VerExp})\) consists of the following algorithms:

\begin{itemize}
    \item \(\mathsf{VerExp.Setup}\): Receive the security parameter \(\lambda\) with circuit \(\mathcal{C}\) and generate public parameter \(\mathsf{pp}\), proving key \(\mathsf{pk}\), and verification key \(\mathsf{vk}\).
    \item \(\mathsf{VerExp.Prove}\): Receive \(\mathsf{pk}\) with the range of the median query \(\mathcal{R}=\{\mathsf{range}_i\}_{i\in[n]}\) and return a differentially private median estimation \(\mathsf{med}\), commitments \(\{\mathsf{com}_i\}_{i\in[m]}\), and proof \(\pi\).
    \item \(\mathsf{VerExp.Verify}\): Receive the following tuple,
    \[
        (\mathsf{vk}, \mathsf{med}, \{\mathsf{range}_i\}_{i\in[n]}, \{\mathsf{c}_i\}_{i\in[m]}, \{\mathsf{com}_i\}_{i\in[m]}, \pi)
    \]
    and return a bit representing \(accept\) or \(reject\). \(\{\mathsf{c}_i\}_{i\in[m]}\) is the commitment published on the public bulletin board.
\end{itemize}

The pipeline of whole scheme consists of four phases.

\begin{table}[h]
  \centering
  \begin{tabular}{p{0.95\linewidth}}
    \toprule
        \(\mathsf{VerExpPipeline}(\{\mathcal{D}_i\}_{i\in[m]}, \mathcal{P}, \mathcal{V}, \lambda)\) \\ 
    \midrule
    \begin{minipage}{\linewidth}
      \begin{algorithmic}[1]
        \STATE \((\mathsf{pp}, \mathsf{pk}, \mathsf{vk})\leftarrow\mathsf{VerExp.Setup}(1^\lambda, \mathcal{C})\)
        \STATE \(\mathsf{pp}\) remains public. \(\mathsf{pk}\) is sent to \(\mathcal{P}\) and \(\mathsf{vk}\) is sent to \(\mathcal{V}\)
        \FOR{\(i\in [m] \) }
            \STATE \(\mathcal{D}_i\) obtains \(r_i\) from its private randomness
            \STATE \(\mathcal{D}_i\) obtains \(\mathsf{c}_i\leftarrow\mathsf{Com}(x_i,r_i)\).
            \STATE \(\mathcal{D}_i\) publishes \(\mathsf{c}_i\) into the public bulletin board
            \STATE \(\mathcal{D}_i\) sends \((x_i, r_i)\) to \(\mathcal{P}\)
        \ENDFOR
        \STATE \(\mathcal{P}\) generates witness \(w \leftarrow \mathsf{WGen}(\mathcal{C}, \{(x_i, r_i)\}_{i\in[m]})\)
        \STATE \(\mathcal{P}\) runs \[\mathsf{VerExp.Prove}(\mathsf{pk}, \{\mathsf{range}_i\}_{i\in[m]}, w)\] and to obtain and publish \((\mathsf{med},\{\mathsf{com}_i\}_{i\in[m]}, \pi)\)
        \STATE \(\mathcal{V}\) reads the public bulletin board and runs 
        \begin{align*}
            \mathsf{VerExp.Verify}(\mathsf{vk}, \mathsf{med}, \{\mathsf{range}_i\}_{i\in[n]}, \\\{\mathsf{c}_i\}_{i\in[m]}, \{\mathsf{com}_i\}_{i\in[m]}, \pi)
        \end{align*}
        to obtain the output \(b\)
        \IF{\(b=accept\)}
            \STATE \(\mathcal{V}\) accepts
        \ELSE
            \STATE\(\mathcal{V}\) rejects
        \ENDIF
      \end{algorithmic}
    \end{minipage}
    \\ \bottomrule 
  \end{tabular}
  \caption{Pipeline of \(\mathsf{VerExp}\)}
  \label{table:pipeline} 
\end{table}

\textbf{Setup Phase}. (1st-2nd row of Table~\ref{table:pipeline}) Setup participants invokes the algorithm \(\mathsf{VerExp.Setup}\) to generate the public parameter, the proving key, and the verification key. The necessity of a trusted setup party depends on the zk-SNARK scheme being used. Then the key is distributed to \(\mathcal{P}\) and \(\mathcal{V}\).

\textbf{Commitment Phase}. (3rd-6th row of Table~\ref{table:pipeline}) Each Data Provider commits to its private input \(x_i\) by publishing \(\mathsf{c}_i=\mathsf{Com}(x_i,r_i)\) on a public bulletin board. We use a hash-based commitment.

\textbf{Proof Generation Phase}. (7th-10th row of Table~\ref{table:pipeline}) Each Data Provider transmits \((x_i,r_i)\) to the Data Analyst over a secure channel. The Data Analyst aggregates these inputs to generate a witness \(w\) with witness generator function \(\mathsf{WGen}\). Then after generating the proof by running the algorithm \(\mathsf{zkSNARK.Prove}\), the Data Analyst publishes the perturbed statistics together with the corresponding zero-knowledge proof produced by the zk-SNARK scheme.

\textbf{Verification Phase}. (11th-16th row of Table~\ref{table:pipeline}) The procedure is as follows. (i) Confirm that the circuit's commitments output (\(\{\mathsf{com}_i\}_{i\in[m]}\)) match those posted on the public bulletin board (\(\{\mathsf{c}_i\}_{i\in[m]}\)). (ii) Check \(\{\mathsf{range}\}_{i\in[n]}\) matches the public inputs of the circuit. (iii) Check the proof's validity. Since all materials for verification are publicly accessible, the scheme publicly verifiable.

\begin{algorithm}[h!]
    \caption{\(\mathsf{VerExp.Setup}\)}
    \label{alg:algorithm1}
    \textbf{Input}: \(1^\lambda, \mathcal{C}\)\\
    \textbf{Output}: \(\mathsf{pp}, \mathsf{pk}, \mathsf{vk}\)
        \begin{algorithmic}[1] 
            \STATE \((\mathsf{pp}_{\mathsf{zkSNARK}}, \mathsf{pk}, \mathsf{vk}) \leftarrow \mathsf{zkSNARK.Setup}(1^\lambda, \mathcal{C})\)
            \STATE \(\mathsf{pp}_\mathsf{Comitment} \leftarrow \mathsf{Gen}(1^\lambda)\)
            \STATE \(\mathsf{pp}=(\mathsf{pp}_{\mathsf{zkSNARK}},\mathsf{pp}_\mathsf{Comitment})\)
            \STATE \textbf{return} \(\mathsf{pp}, \mathsf{pk}, \mathsf{vk}\)
        \end{algorithmic}
\end{algorithm}

\begin{algorithm}[h!]
    \caption{\(\mathsf{VerExp.Prove}\)}
    \label{alg:algorithm2}
    \textbf{Input}: \(\mathsf{pk}, \{\mathsf{range}_i\}_{i\in[n]}, w\)\\
    \textbf{Output}: \(\mathsf{med}, \{\mathsf{com}_i\}_{i\in[m]}, \pi\)
        \begin{algorithmic}[1] 
            \STATE \((Output, \pi) \leftarrow \mathsf{zkSNARK.Prove}( \mathsf{pk}, \{\mathsf{range}_i\}_{i\in[n]}, w)\)
            \STATE Parse \(Output=\mathsf{med} \vert\vert (\mathsf{com}_1,...,\mathsf{com}_m)\)
            \STATE \textbf{return} \((\mathsf{med}, \{\mathsf{com}_i\}_{i\in[m]}, \pi)\)
        \end{algorithmic}
\end{algorithm}

\begin{algorithm}[h!]
    \caption{\(\mathsf{VerExp.Verify}\)}
    \label{alg:algorithm3}
    \textbf{Input}: \(\mathsf{vk}, \mathsf{med}, \{\mathsf{range}_i\}_{i\in[n]}, \{\mathsf{c}_i\}_{i\in[m]}, \{\mathsf{com}_i\}_{i\in[m]}, \pi\)\\
    \textbf{Output}: \(b'\in\{accept,reject\}\)
        \begin{algorithmic}[1] 
            \STATE  \(b'\leftarrow \mathsf{zkSNARK.Verify}(\mathsf{vk}, \{\mathsf{range}_i\}_{i\in[n]}, \mathsf{med}, \pi)\)
            \IF{\(b'=accept\) \AND \(\forall i\in[m], \mathsf{c}_i=\mathsf{com}_i\)}
                \STATE \textbf{return} \(accept\)
            \ELSE
                \STATE \textbf{return} \(reject\)
            \ENDIF
        \end{algorithmic}
\end{algorithm}

\subsection{Main Circuit $\mathcal{C}$}

Main circuit \(\mathcal{C}\) is composed of six modules: \(\mathsf{Util}\), \(\mathsf{SubMin}\), \(\mathsf{ExpLookup}\), \(\mathsf{InverseCDF}\), \(\mathsf{Bind}\), and \(\mathsf{Mod}\). 

\textbf{Inputs}: 
\begin{itemize}
    \item \(\{\mathsf{range}_i\}_{i\in[n]}\) : Range of median query
    \item \(\{\mathsf{input}_i\}_{i\in[m]}\) : Inputs from Data Providers
    \item \(\{\mathsf{rand}_i\}_{i\in[m]}\) : Randomness from Data Providers
\end{itemize}

\textbf{Outputs}:
\begin{itemize}
    \item $\{\mathsf{com}_i\}_{i\in[m]}$ : Commitments to Data Providers' inputs
    \item $\mathsf{med}$ : DP median estimation 
\end{itemize}
We assume \(\{\mathsf{range}_i\}_{i\in[n]}\) and all output are public. \(\{\mathsf{input}_i\}_{i\in[m]}\) and \(\{\mathsf{rand}_i\}_{i\in[m]}\) are private.

\begin{table}[h]
  \centering
  \begin{tabular}{p{0.95\linewidth}}
    \toprule
        \(\mathcal{C}(\{\mathsf{range}_i\}_{i\in[n]}, \{\mathsf{input}_i\}_{i\in[m]}, \{\mathsf{rand}_i\}_{i\in[m]})\) \\ 
    \midrule
    \begin{minipage}{\linewidth}
      \begin{algorithmic}[1]
        \STATE \(\{\mathsf{com}_i\}_{i\in[m]}\leftarrow\mathsf{Bind}(\{\mathsf{input}_i\}_{i\in[m]}, \{\mathsf{rand}_i\}_{i\in[m]})\)
        \STATE \(\{\mathsf{util}'_i\}_{i\in[n]}\leftarrow\mathsf{Util}(\{\mathsf{range}_i\}_{i\in[n]},\{\mathsf{input}_i\}_{i\in[m]})\)
        \STATE \(\mathsf{expval}_0\leftarrow\mathsf{ExpLookup}(\mathsf{util}'_0)\)
        \STATE \(s_0\leftarrow\mathsf{expval}_0\)
        \FOR{\(1\le i \le n-1\)}
            \STATE \(\mathsf{expval}_i\leftarrow\mathsf{ExpLookup}(\mathsf{util}'_i)\)
            \STATE \(s_i\leftarrow s_{i-1}+\mathsf{expval}_i\)
        \ENDFOR
        \STATE \(\rho\leftarrow \mathsf{Mod}(\sum_{i=0}^{m-1}\mathsf{rand}_i)\)
        \STATE \(\mathsf{med} \leftarrow \mathsf{InverseCDF}(\{s_i\}_{i\in[n]}, \{\mathsf{range}_i\}_{i=0}^{n-1}, \rho)\)
        \STATE \textbf{output} \((\mathsf{med}, \{\mathsf{com}_i\}_{i\in[m]})\)
      \end{algorithmic}
    \end{minipage}
    \\ \bottomrule
  \end{tabular}
  \caption{Main Circuit \(\mathcal{C}\) for Verifiable Exponential Mechanism}
  \label{table:Main Circit} 
\end{table}

\subsubsection*{Implementation Challenges and Design Decisions} Implementing the exponential mechanism inside a static arithmetic circuit presents two main obstacles. First, the circuit must perform verifiable sampling according to a distribution derived from the utility function, which is nontrivial. Second, standard stochastic methods like rejection sampling rely on data‐dependent control flow or dynamic branching, which is incompatible with the fixed‐depth nature of static circuits.

To overcome these issues, we employ a scaled version of inverse cumulative distribution function (CDF) sampling. By preparing a lookup table of scaled exponential values in advance and encoding the sampling process as a small, fixed‐depth sequence of comparisons and multiplications, we avoid any dynamic branching. This design ensures (i) the sampling faithfully follows the distribution, and (ii) all operations remain within the static circuit model.

Since arithmetic circuits for zk-SNARKs are defined over a finite field, infinite accuracy on real-valued exponential values is impossible. Therefore, we propose a conceptual method to divide the lookup table into two parts. The first part is a physical lookup table $t$. This is the area that works like a normal lookup table that approximates relatively ``large'' exponentials that can be approximated into an integer without disrupting the ratio between its adjacent entries too much. The other part is filled with all 0 (method \textbf{\textit{set0})} or all $k=\left\lceil\frac{1}{\exp(\varepsilon/2)-1}\right\rceil$ (method \textbf{\textit{setk}}). We call the whole infinite-sized lookup table (\(t\) and the rest) as the conceptual lookup table $T$.

\subsection{Implementation Details}

In this section, we provide the details of the modules that constitute the main circuit. Pseudocodes and figures intended to help the reader's understanding are provided in Appendix~\ref{section:pseudocodes} and Appendix~\ref{section:figures}.

\textbf{\(\mathsf{Bind}\) Module} receives \(\{(x_i,r_i)\}_{i\in[m]}\) from each \(\mathcal{D}_i\) \((i\in[m])\) and outputs $m$ commitments \(\mathsf{com}_i=\mathsf{hash}(x_i,r_i)\) that bind each provider's data using a cryptographic hash function \(\mathsf{hash}:\{0,1\}^*\rightarrow\mathbb{F}_p\), where \(\mathbb{F}_p\) is a finite field that the arithmetic circuit is defined over. During verification, the Verifier confirms that each output of this module is consistent with the commitment that is posted on the public bulletin board. \(\mathsf{Bind}\) module prevents a malicious Data Analyst from substitute arbitrary values for some \(x_i\) and still passing the verification procedure.

\textbf{\(\mathsf{Util}\) and \(\mathsf{SubMin}\) Module} computes the calibrated utility score of \(\mathsf{range}_i\) for \(i\in[n]\). The module implements the utility function given in Definition ~\ref{def:util and rank} and does some calibration to make the highest utility score be always 0. We implement the $\mathsf{Util}$ module with $m \cdot n$ comparators and $n$ \(\mathsf{Abs}\) submodule that calculates the absolute value of the input. After that, each absolute utility score goes into the \(\mathsf{SubMin}\) module, which receives $n$ absolute values of utility scores as the inputs, calculates the minimum of them, and subtracts the minimum value from each inputs to output the final calibrated utility score.

\textbf{\(\mathsf{ExpLookup}\) Module} act as a lookup table for exponentiations. Since our circuit is defined on a finite field \(\mathbb{F}_p\), infinite accuracy on real-valued exponentials is impossible. Therefore, we propose a conceptual method to divide the lookup table into two parts. The first part is a physical lookup table $t$. This is the area that works like a normal lookup table that approximates relatively ``large'' exponentials that can be approximated into an integer modulo $p$. The other part is filled with all 0 (method \textbf{\textit{set0}}) or all $k=\left\lceil\frac{1}{\exp(\varepsilon/2)-1}\right\rceil$ (method \textbf{\textit{setk}}). We call the whole infinite-sized lookup table as the conceptual lookup table $T$.

Let the size of the physical lookup table $t$ be $l$, which means the first $l$ entries of \(T\) are the same as the physical lookup table \(t\). Then, each entry of the conceptual lookup table is filled as follows.
\[
    T[i] =
    \begin{cases}
    \left\lfloor\exp(\frac{\varepsilon}{2})\cdot T[i+1]\right\rfloor & \text{if }i < l-1\\
    k & \text{if }i=l-1\\
    0 & \text{if }i \ge l \text{ and \textbf{\textit{set0}}} \\
    k & \text{if }i \ge l \text{ and \textbf{\textit{setk}}} \\
    \end{cases}
\]
By the definition of $T[i]$, the ratio of each adjacent entries does not exceed $\exp(\varepsilon/2)$ in the physical lookup table, and this property guarantees differential privacy.

\textbf{\(\mathsf{Mod}\) Module} receives sum of randomnesses from the Data Providers and $s_{n-1}$, and outputs \(\rho=\sum_{i\in[n]}r_i \mod s_{n-1}\) where each \(r_i\) is the private randomness of \(\mathcal{D}_i\). The output is guaranteed to be computationally indistinguishable from uniform distribution over $\mathbb{Z}_{s_{n-1}}$ as long as at least one of the Data Providers have sampled their internal randomness uniformly random. The proof of computational indistinguishability is provided in Lemma~\ref{lem:randomness rho is comp}.

\textbf{\(\mathsf{InverseCDF}\) Module} receives cumulative exponential weights \(\{s_i\}_{i\in[n]}\), range of the median query \(\{\mathsf{range}_i\}_{i\in[m]}\), and a randomness seed $\rho$. It returns a differentially private median estimation by performing a scaled version of inverse‐CDF sampling with randomness \(\rho\) which is computationally indistinguishable from uniformly random modulo \(s_{n-1}\).

\section{Security and Utility Analysis}

In this section, we analyze the security and privacy properties, and utility of our proposed mechanism, and provide formal proofs.  The proofs for the theorems and lemmas in this section is provided in the Appendix.

\subsection{Security and Privacy Analysis}\label{sec:Security and Privacy Analysis}

\textbf{Security Definitions and Guarantees}

Let \(\mathsf{VerExp}\left[\{\mathcal{D}_i\}_{i\in[m]}, \mathcal{P},\mathcal{V}\right]\in\{0,1\}\) denote the decision of the verifier \(\mathcal{V}\) when the three parties \(\{\mathcal{D}_i\}_{i\in[m]}\), \(\mathcal{P}\), and \(\mathcal{V}\) participate in the protocol via the pipeline \(\mathsf{VerExpPipeline}\). 0 denotes \(\mathcal{V}\) rejects, and 1 denotes \(\mathcal{V}\) accepts. By ``honest'' means that the party correctly follows the actions described in the protocol. On the other hand, ``malicious'' means that the party does not follow the actions prescribed in the protocol.

\(\mathsf{VerExp}\) satisfies perfect completeness, negligible soundness, and computational zero-knowledge. We propose formal definitions of these properties as below.

\begin{definition}[Perfect Completeness]
    Any honest \(\mathcal{P}\) that correctly followed the scheme generates a proof that always passes the verification.
    Formally,
    \[
        \Pr
        \begin{bmatrix}
            \mathsf{VerExp}\left[\{\mathcal{D}_i\}_{i\in[m]}, \mathcal{P},\mathcal{V}\right]=0 \\
        \end{bmatrix}
        =0
    \]
     where the probability is taken over the randomness of participants.
\end{definition}

\begin{definition}[Negligible Soundness]
    No malicious \(\mathcal{P}^*\) can generate a proof that convince the Verifier except with negligible probability.
    Formally, 
    \[
        \Pr
        \begin{bmatrix}
            \mathsf{VerExp}\left[\{\mathcal{D}_i\}_{i\in[m]}, \mathcal{P}^*,\mathcal{V}\right]=1 \\
        \end{bmatrix}
        \le \mathsf{negl}(\lambda)
    \]
    for a negligible function \(\mathsf{negl}\) and the security parameter \(\lambda\). The probability is taken over the internal randomness of all participants.
\end{definition}

\begin{definition}[Computational Zero-knowledge]
    Upon successful verification, \(\mathcal{V}\) does not learn anything other than the fact that the differentially private median was sampled correctly, and only negligible amount of private information from \(\{\mathcal{D}_i\}_{i\in[m]}\). Formally,
    For every verifier \(\mathcal{V}\), there exists two PPT simulators \(\mathsf{Sim}_{\mathcal{P}\leftrightarrow\mathcal{V}}\) and \(\mathsf{Sim}_{\mathcal{B}\leftrightarrow\mathcal{V}}\) on input \(X\) such that
    \begin{align*}
        \mathsf{Sim}_{\mathcal{P}\leftrightarrow\mathcal{V}}(X)&\overset{c}{\approx} \mathsf{View}_{\mathcal{V}}\left[\mathcal{P}(X,W)\leftrightarrow\mathcal{V}(X)\right] \\
        \mathsf{Sim}_{\mathcal{B}\leftrightarrow\mathcal{V}}(X)&\overset{c}{\approx} \mathsf{View}_{\mathcal{V}}\left[\mathcal{B}(\mathsf{pp}) \leftrightarrow\mathcal{V}(X)\right]
    \end{align*}
    For simplicity of notation, we denote \(X\) and \(W\) as follows.
    \begin{align*}
        X&=(\mathsf{pp}, \mathsf{vk}, \{\mathsf{range}_i\}_{i\in[n]})\\
        W&=(\mathsf{pk}, \{x_i\}_{i\in[m]}, \{r_i\}_{i\in[m]})
    \end{align*}
\end{definition}

\begin{theorem} \label{thm:comp, sound, zk}
    \(\mathsf{VerExp}\) satisfies perfect completeness, negligible soundness, and computational zero-knowledge.
\end{theorem}

\begin{proof} 
    \textbf{Perfect Completeness}.    
    Assume that \[(\mathsf{med}, \{\mathsf{com}_i\}_{i\in[m]}, \pi)\] was generated honestly by using correct proving key \(\mathsf{pk}\), correct range of the median estimation \(\{\mathsf{range}_i\}_{i\in[n]}\), and a correct witness \(w\) that was honestly generated with inputs \((x_i,r_i)_{i\in[m]}\) from the Data Providers. 
    
    (i) Assuming that the underlying zkSNARK scheme has perfect completeness, the following holds:
    \begin{align*}
        \Pr[accept\leftarrow\mathsf{zkSNARK.Verify}(\mathsf{vk}, \{\mathsf{range}_i\}_{i\in[n]},\mathsf{med}, \pi)]=1
    \end{align*}

    (ii) By the design of our main circuit (specifically the \(\mathsf{Bind}\) module), if the prover honestly used correct inputs \((x_i,r_i)_{i\in[m]}\), the resulting output of the circuit \(\{\mathsf{com}_i\}_{i\in[m]}\) is always identical to the commitments \(\{\mathsf{c}_i\}_{i\in[m]}\) from the Public Bulletin Board.

    Combining (i) and (ii) we get the following.
    \begin{align*}
        \Pr[&accept\leftarrow\mathsf{VerExp.Verify}(\mathsf{vk},\mathsf{med} , \{\mathsf{range}_i\}_{i\in[n]}, \\
        &\;\;\;\;\;\;\;\;\;\;\;\;\;\;\;\;\;\;\;\;\;\;\;\;\;\;\;\;\;\;\;\;\;\;\;
        \{\mathsf{c}_i\}_{i\in[m]}, \{\mathsf{com}_i\}_{i\in[m]}, \pi)]=1
    \end{align*}

    \textbf{Negligible Soundness}. Under the assumption that the underlying zk-SNARKs protocol provides negligible soundness, a malicious \(\mathcal{P}^*\) could cheat by generating a false proof only by doing at least one of the followings:
    \begin{itemize}
        \item Tamper with the inputs from the Data Providers.
        \item Tamper with the range of the query.
    \end{itemize}
    
    The former is computationally infeasible under the assumption of second-preimage resistance of the cryptographic hash function; \(\mathcal{P}^*\) must find another input for the hash that produces the same hash value to pass the verification \(\mathsf{VerExp.Verify}\). 
    
    The latter depends on the soundness of the underlying zk-SNARK protocol; the probability to success such cheating depends on the probability that we use forged public input and succeeds to pass the verification \(\mathsf{zkSNARK.Verify}\). zkSNARK schemes (including Groth16) depend on Schwarts-Zippel lemma to achieve very short proof and verification time. However, the method implies negligible, but nonzero soundness error. Therefore the latter method of cheating is computationally infeasible yet not impossible. Overall, the probability that \(\mathcal{P}^*\) succeeds cheating is negligible.
    
    \textbf{Computational Zero-knowledge}. First, since all existing zk-SNARK protocols depend on cryptographic assumptions, the verification key  \(\mathsf{vk}\) and the proof \(\pi\) is computationally zero-knowledge.
    Also, under random oracle model, the distribution of the commitment \(\mathsf{com}_i\) to private data \(x_i\) from \(\mathcal{D}_i\) is uniform. Therefore, no information about \(x_i\) can be leaked here. (Of course, if we only use the assumption of cryptographic hash, \(\{\mathsf{com}_i\}_{i\in[m]}\) would be only computationally zero-knowledge.) 

    For a formal proof, we can construct a PPT simulator \(\mathsf{Sim}_{\mathcal{P}\leftrightarrow\mathcal{V}}\) and \(\mathsf{Sim}_{\mathcal{B}\leftrightarrow\mathcal{V}}\) under random oracle model as Table~\ref{table:Sim_PV} and Table~\ref{table:Sim_BV}. Recall that the public input and public output are \[\{\mathsf{range}_i\}_{i\in[n]}, \{\mathsf{com}_i\}_{i\in[m]}, \text{ and } \mathsf{med}\]

    \begin{table}[h]
      \centering
      \begin{tabular}{p{0.95\linewidth}}
        \toprule
            \(\mathsf{Sim}_{\mathcal{B}\leftrightarrow\mathcal{V}}(\{\mathsf{range}_i\}_{i\in[n]}, \{\mathsf{com}_i\}_{i\in[m]}, \mathsf{med})\) \\ 
        \midrule
        \begin{minipage}{\linewidth}
          \begin{algorithmic}[1]
            \FOR{\(i\in[m]\)}
                \IF{\(0\le i \le m-2\)}
                    \STATE Sample \(x_i\) from \(\{\mathsf{range}_i\}_{i\in[n]}\)
                    \STATE Sample \(r_i \overset{\$}{\leftarrow} \mathcal{U}(\mathbb{F}_p)\)
                    \STATE Program \(\mathsf{hash}(x_i||r_i)=\mathsf{com}_i\)
                \ELSE
                    \STATE Sample \(x_{m-1}\) from \(\{\mathsf{range}_i\}_{i\in[n]}\)
                    \STATE \(\{\mathsf{util}'_i\}_{i\in[n]}\leftarrow\mathsf{Util}(\{\mathsf{range}_i\}_{i\in[n]},\{x_i\}_{i\in[m]})\)
                    \STATE \(\mathsf{expval}_0\leftarrow\mathsf{ExpLookup}(\mathsf{util}'_0)\)
                    \STATE \(s_0\leftarrow\mathsf{expval}_0\)
                    \FOR{\(1\le i \le n-1\)}
                        \STATE \(\mathsf{expval}_i\leftarrow\mathsf{ExpLookup}(\mathsf{util}'_i)\)
                        \STATE \(s_i\leftarrow s_{i-1}+\mathsf{expval}_i\)
                    \ENDFOR
                    \STATE Find \(\alpha\) such that \(\mathsf{range}_\alpha=\mathsf{med}\)
                    \STATE Set \(r_{m-1}\) such that 
                    \[s_{\alpha}\le\left(\sum_{j=0}^{m-1}r_j \mod{p}\right) \mod{s_{n-1}}\le s_{\alpha+1}\]
                    \STATE Program \(\mathsf{hash}(x_{m-1}||r_{m-1})=\mathsf{com}_{m-1}\)
                \ENDIF
            \ENDFOR
            \STATE Send \(\{\mathsf{com}_i\}_{i\in[m]}\) to \(\mathcal{V}\)
          \end{algorithmic}
        \end{minipage}
        \\ \bottomrule
      \end{tabular}
      \caption{Simulator \(\mathsf{Sim}_{\mathcal{B}\leftrightarrow\mathcal{V}}\)}
      \label{table:Sim_BV} 
    \end{table}

\FloatBarrier

    \begin{table}[h]
      \centering
      \begin{tabular}{p{0.95\linewidth}}
        \toprule
            \(\mathsf{Sim}_{\mathcal{P}\leftrightarrow\mathcal{V}}(\{\mathsf{range}_i\}_{i\in[n]}, \{\mathsf{com}_i\}_{i\in[m]}, \mathsf{med})\) \\ 
        \midrule
        \begin{minipage}{\linewidth}
          \begin{algorithmic}[1]
            \STATE Receive \(\{x_i,r_i\}_{i\in[m]}\) from \(\mathsf{Sim}_{\mathcal{B}\leftrightarrow\mathcal{V}}\)
            \STATE \(w \leftarrow \mathsf{WGen}(\mathcal{C}, \{(x_i,r_i)\}_{i\in[m]})\)
            \STATE Run \(\mathsf{VerExp.Prove}(\mathsf{pk}, \{\mathsf{range}_i\}\) to obtain
            \[
                (\mathsf{med}, \{\mathsf{com}_i\}_{i\in[m]}, \pi)
            \]
            \STATE Send \((\mathsf{med}, \{\mathsf{com}_i\}_{i\in[m]}, \pi)\) to \(\mathcal{V}\)
          \end{algorithmic}
        \end{minipage}
        \\ \bottomrule  
      \end{tabular}
      \caption{Simulator \(\mathsf{Sim}_{\mathcal{P}\leftrightarrow\mathcal{V}}\)}
      \label{table:Sim_PV} 
    \end{table}

\end{proof}

\textbf{Privacy Guarantees}

We now prove that each method for filling the lookup table ensures that the resulting median estimation satisfies computational differential privacy.

$T$ is the conceptual lookup table. We define \[\text{OPT}_u(\mathcal{DB})= \text{max}_{r\in\mathcal{R}}u(\mathcal{DB},r)\] to be the maximum possible utility score from a database $\mathcal{DB}$. Let $U_\mathcal{DB}$ be the table of the exponentiated and calibrated utility scores of each element in the range of the query under a database $\mathcal{DB}$. i.e. $U_\mathcal{DB}$ is the table of un-normalized probability that each element is picked as a DP median. We assume $\sum_{r \in \mathcal{R}}U_\mathcal{DB}[r]=N$. Note that each entry in $U_\mathcal{DB}$ is calculated as follows. 
\begin{align*}
    U_\mathcal{DB}[r] = T[\text{OPT}_u(\mathcal{DB})-u(\mathcal{DB},r)]
\end{align*}

\begin{theorem} \label{thm:set0_approxDP}
    \textbf{\textit{set0}} satisfies $(\varepsilon,\delta)$-computational DP for $\delta=e^{\frac{\varepsilon}{2}}\cdot\frac{k}{N}+\frac{s_{n-1}}{2p}$, where $k=\left\lceil\frac{1}{\exp(\varepsilon/2)-1}\right\rceil$.
\end{theorem}

\begin{proof}
    We first prove the differential privacy of \textbf{\textit{set0}} with assumption that the randomness \(\rho\) is uniform over \(\mathbb{Z}_{s_{n-1}}\). Then combining Lemma~\ref{lem:randomness rho is comp} and Lemma~\ref{lem:randomness change}, we obtain the proof of computational differential privacy.
    
    Let $l$ be the size of the ``physical'' lookup table, which is an integer such that $T[i]=0$ for $i \ge l$ and $T[i] > 0$ for $i < l$. Let $\mathsf{range}$ be the array of the all the elements in range of the median estimation. $\mathsf{range}[a:b]$ is the subarray of $\mathsf{range}$ that includes from $\mathsf{range}[a]$ to $\mathsf{range}[b]$. Let $\alpha, \beta\in [n]$ be indices of the array $\mathsf{range}$ that satisfy the following condition.
    \begin{align*}
        \forall j\in [0,n)&,\\ j \le &\alpha \text{ or } j \ge \beta \implies u(\mathcal{DB}, \mathsf{range}[j]) \le -l
    \end{align*}
    Let $\mathcal{DB}$ and $\mathcal{DB}'$ be two neighboring databases and $\Pr[r,\mathcal{DB}]$ be the probability that our mechanism run over the database \(\mathcal{D}\) outputs $r$. \(U_{\mathcal{DB}}\) is defined in Section~\ref{sec:Security and Privacy Analysis}.
    
    (i) For $r'\in \mathsf{range}[\alpha+2:\beta-2]$,
    \begin{align*}
        \frac{\Pr[r', \mathcal{DB}]}{\Pr[r', \mathcal{D}']} &= \frac{\frac{U_\mathcal{DB}[r']}{\sum_{r \in \mathcal{R}}U_\mathcal{DB}[r]}}{\frac{U_\mathcal{DB'}[r']}{\sum_{r\in\mathcal{R}}U_\mathcal{DB'}[r]}} \\
        &= \frac{U_\mathcal{DB}[r']}{U_\mathcal{DB'}[r']} \cdot \frac{\sum_{r\in\mathcal{R}}U_\mathcal{DB}[r]}{\sum_{r\in\mathcal{R}}U_\mathcal{DB'}[r]} \\
        &\le e^{\frac{\varepsilon}{2}} \cdot e^{\frac{\varepsilon}{2}} \\
        &= e^{\varepsilon}
    \end{align*}
    By symmetry, we can get the following inequality.
    \begin{align*}
        \frac{\Pr[r', \mathcal{DB}]}{\Pr[r', \mathcal{DB}']} &\ge e^{-\varepsilon} \text{ for }r'\in \mathsf{range}[\alpha+2:\beta-2]
    \end{align*}
   
    (ii) For $r'\in \mathsf{range}[0:\alpha-1]\cup \mathsf{range}[\beta+1:n-1]$,
    \begin{align*}
        \Pr[r',\mathcal{DB}]&=\frac{U_\mathcal{DB}[r']}{\sum_{r\in \mathcal{R}}U_\mathcal{DB}[r]}=0
        \\
        \Pr[r',\mathcal{DB}']&=\frac{U_\mathcal{DB'}[r']}{\sum_{r\in \mathcal{R}}U_\mathcal{DB'}[r]}=0
    \end{align*}
    since $U_\mathcal{DB}[r']=U_\mathcal{DB'}[r']=0$.
        
    (iii) For $r' \in \mathsf{range}[\alpha:\alpha+1]\cup \mathsf{range}[\beta:\beta-1]$,
    \begin{align*}
        &U_\mathcal{DB'}[r'] \in [e^{-\frac{\epsilon}{2}}U_\mathcal{DB}[r'], e^{\frac{\varepsilon}{2}}U_{\mathcal{DB}}[r']]\;\;\;\text{ or } \\
        &U_\mathcal{DB'}[r'] \in [U_\mathcal{DB}[r']-k, U_\mathcal{DB}[r']+k]
    \end{align*}
    since $U_\mathcal{DB}(\mathsf{range}[\alpha])=U_\mathcal{DB}(\mathsf{range}[\beta])=0$ while $U_\mathcal{DB}(\mathsf{range}[\alpha+1])=U_\mathcal{DB}(\mathsf{range}[\beta-1])=k$.
    
    Combining above, we get the following result.
    \begin{align*}
        \forall r\in \mathcal{R},\Pr[r,\mathcal{DB}] \le e^{\varepsilon}\Pr&[r, \mathcal{DB}']+\delta \\ &\text{ for }\delta \le e^{\frac{\varepsilon}{2}} \cdot \frac{k}{N}
    \end{align*}
\end{proof}

\begin{theorem} \label{thm:setk_pureDP}
    \textbf{\textit{setk}} satisfies $(\varepsilon, \frac{s_{n-1}}{2p})$-computational DP.
\end{theorem}

\begin{proof}
    We first prove the differential privacy of \textbf{\textit{setk}} with assumption that the randomness \(\rho\) is uniform over \(\mathbb{Z}_{s_{n-1}}\). Then combining Lemma~\ref{lem:randomness rho is comp} and Lemma~\ref{lem:randomness change}, we obtain the proof of computational differential privacy.
    
    All adjacent entries in $U_\mathcal{DB}$ differ by ratios bounded above by $\exp(\varepsilon/2)$ and the sensitivity of the utility function is $1$. Therefore, for all $r'\in\mathcal{R}$ and neighboring databases $\mathcal{DB}, \mathcal{DB}'$, the following holds.
    \begin{align*}
        U_\mathcal{DB'}[r'] \in [e^{-\frac{\epsilon}{2}}U_\mathcal{DB}[r'], e^{\frac{\varepsilon}{2}}U_{\mathcal{DB}}[r']]
    \end{align*}
    By leveraging this property, we get
    \begin{align*}
        \frac{\Pr[\mathsf{range}[i], \mathcal{DB}]}{\Pr[\mathsf{range}[i], \mathcal{DB}']} &= \frac{\frac{U_\mathcal{DB}[i]}{\sum_{j=0}^{k-1}U_\mathcal{DB}[j]}}{\frac{U_\mathcal{DB'}[i]}{\sum_{j=0}^{k-1}U_\mathcal{DB'}[j]}} \\
        &= \frac{U_\mathcal{DB}[i]}{U_\mathcal{DB'}[i]} \cdot \frac{\sum_{j=0}^{k-1}U_\mathcal{DB}[j]}{\sum_{j=0}^{k-1}U_\mathcal{DB'}[j]} \\
        &\le e^{\frac{\varepsilon}{2}} \cdot e^{\frac{\varepsilon}{2}} \\
        &= e^{\varepsilon}
    \end{align*}
    By symmetry, we can also get the following inequality.
    \begin{align*}
        \frac{\Pr[\mathsf{range}[i], \mathcal{DB}]}{\Pr[\mathsf{range}[i], \mathcal{DB}']} &\ge e^{-\varepsilon}
    \end{align*}
    Finally, we get the following result.
    \begin{align*}
        \Pr[\text{range}[i],\mathcal{DB}] \le e^{\varepsilon}\Pr[\mathsf{range}[i], \mathcal{DB}']
    \end{align*}
\end{proof}

\subsection{Utility Analysis}

In this section, we provide the utility analysis of our scheme. We bound the probability of returning a bad estimation (i.e. an estimation that is far from the true median) and show that such probability is exponentially small. Let $\mathcal{R}_{\text{OPT}}=\{r\in\mathcal{R}:u(\mathcal{DB},r)=\text{OPT}_u(\mathcal{DB})\}$. 

First, we model the entries of the lookup table as a sequence $A[i]$ and the true exponentially decaying sequence as $B[i]$. $B[i]=N \cdot \exp(-\frac{\varepsilon}{2}\cdot i)$ and $A[i]$ is filled the same way the lookup table is filled. Ideally, $A[i]$ should be equal to $B[i]$, which is an exactly exponentially decaying sequence. However, since $A[i]$ is an integer approximation of exponential decay such that each adjacent elements should only differ up to the multiplicative factor of $e$, we should bound the error derived from this approximation.

Now, we provide the utility proof of \textbf{\textit{set0}} and \textbf{\textit{setk}}.

\begin{theorem} \label{thm:set0_utility}
    (Utility of \textbf{\textit{set0}}) Let $x\in\mathcal{R}$ be the output of our scheme on database $\mathcal{D}$. Let the conceptual lookup table be $T$. Assume that $l$ is the first index $i$ such that $T[i]=0$. \\
    If \(c>-l\),
    {
        \begin{align*}
            \Pr[u&(\mathcal{DB},x)\le c] \\ \le& \frac{\left|\mathcal{R}\right|}{\left|\mathcal{R}_{\textup{OPT}}\right|} \left( \exp \left(\frac{\varepsilon\left(c-\textup{OPT}_u(\mathcal{DB})\right)}{2}\right)+\frac{e^{\varepsilon/2}}{N (e^{\varepsilon/2}-1)} \right)
            + \frac{s_{n-1}}{4p}
        \end{align*}
    }
    If \(c\le -l\), 
    \begin{align*}
        \Pr[u(\mathcal{DB},x)\le c] = 0
    \end{align*}
\end{theorem}

\begin{proof}
    Similar to the method of proof of Theorem~\ref{thm:set0_approxDP} and Theorem~\ref{thm:setk_pureDP}, we first prove the theorem under the assumption that the randomness \(\rho\) is uniformly random. Then combining Lemma~\ref{lem:randomness rho is comp} and Lemmma~\ref{lem:randomness switch bad prob}, we get the utility proof for our computational DP scheme.
    
    As defined in the proof of Theorem~\ref{thm:set0_approxDP}, \(l\) denotes the size of the ``physical'' lookup table. For $c>-l$ and $r\in \mathcal{R}$ such that \(u(\mathcal{DB},r)\le c\), the unnormalized probability mass that our scheme outputs $r$ is at most 
    \[N\exp\left(\frac{\varepsilon (c-\text{OPT}_u(\mathcal{DB}))}{2}\right)+\frac{e^{\varepsilon/2}}{e^{\varepsilon/2}-1}\] 
    since the approximation of the scaled exponential values are accurate up to additive error of \(\frac{e^{\varepsilon/2}}{e^{\varepsilon/2}-1}\) until $T[l-1]$ due to Lemma~\ref{lem: approx err}.
    
    There can be at most $|\mathcal{R}|$ ``bad'' elements that lies beyond the utility score of $c$. Moreover, by the design of the \(\mathsf{SubMin}\) module, the highest calibrated utility score is always 0. Therefore by the definition of $|\mathcal{R}_{\text{OPT}}|$, there are exactly $|\mathcal{R}_{\text{OPT}}|$ elements in $\mathcal{R}$ that have an unnormalized probability mass of $N$. Combining these facts leads to the following inequality.
    \begin{align*}
        \Pr[u&(\mathcal{DB},x)\le c] \\ 
        &\le \frac{\left|\mathcal{R}\right|\left(N\exp\left(\frac{\varepsilon (c-\text{OPT}_u(\mathcal{DB}))}{2}\right) + \frac{e^{\varepsilon/2}}{e^{\varepsilon/2}-1}\right)} {N\left|\mathcal{R}_{\text{OPT}}\right|} \\
        &= \frac{\left|\mathcal{R}\right|}{\left|\mathcal{R}_{\text{OPT}}\right|}  \left( \exp \left(\frac{\varepsilon\left(c-\text{OPT}_u(\mathcal{DB})\right)}{2}\right)+\frac{e^{\varepsilon/2}}{N (e^{\varepsilon/2}-1)} \right)
    \end{align*}
    The second inequality follows from the fact that the entries of the conceptual lookup table after $T[l]$ are set to $0$.
\end{proof}

\begin{theorem} \label{thm:setk_utility}
    (Utility of \textbf{\textit{setk}}) Let $x\in \mathcal{R}$ be the output of our scheme. Let the conceptual lookup table be $T$. Then,
    \begin{align*}
        \Pr[u&(\mathcal{DB},x)\le c] \\ \le& \frac{\left|\mathcal{R}\right|}{\left|\mathcal{R}_{\textup{OPT}}\right|} \left( \exp \left(\frac{\varepsilon\left(c-\textup{OPT}_u(\mathcal{DB})\right)}{2}\right)+\frac{e^{\varepsilon/2}}{N(e^{\varepsilon/2}-1)} \right)
        + \frac{s_{n-1}}{4p}
    \end{align*}
    
\end{theorem}

\begin{proof}
    Similar to the method of proof of Theorem~\ref{thm:set0_approxDP} and Theorem~\ref{thm:setk_pureDP}, we first prove the theorem under the assumption that the randomness \(\rho\) is uniformly random. Then combining Lemma~\ref{lem:randomness rho is comp} and Lemmma~\ref{lem:randomness switch bad prob}, we get the utility proof for our computational DP scheme.

    For $r\in \mathcal{R}$ such that \(u(\mathcal{DB},r)\le c\), the unnormalized probability mass that our scheme outputs $r$ is at most \[N\exp\left(\frac{\varepsilon (c-\text{OPT}_u(\mathcal{D}))}{2}\right)+\frac{e^{\varepsilon/2}}{e^{\varepsilon/2}-1}\] due to Lemma~\ref{lem: approx err}. 
    
    There can be at most $|\mathcal{R}|$ ``bad'' elements that lies beyond the utility score of $c$. Moreover, by the design of the \(\mathsf{SubMin}\) module, the highest calibrated utility score is always 0. By the definition of $|\mathcal{R}_{\text{OPT}}|$, there are exactly $|\mathcal{R}_{\text{OPT}}|$ elements in $\mathcal{R}$ that possesses an unnormalized probability mass of $N$. Combining these facts leads to the following inequality.
    \begin{align*}
        \Pr[u&(\mathcal{DB},x)\le c] \\ 
        &\le \frac{\left|\mathcal{R}\right|\left(N\exp\left(\frac{\varepsilon (c-\text{OPT}_u(\mathcal{DB}))}{2}\right) +\frac{e^{\varepsilon/2}}{e^{\varepsilon/2}-1}\right)} {N\left|\mathcal{R}_{\text{OPT}}\right|} \\
        &= \frac{\left|\mathcal{R}\right|}{\left|\mathcal{R}_{\text{OPT}}\right|} \left( \exp \left(\frac{\varepsilon\left(c-\text{OPT}_u(\mathcal{DB})\right)}{2}\right)+\frac{e^{\varepsilon/2}}{N (e^{\varepsilon/2}-1)} \right)
    \end{align*}
\end{proof}

\section{Experiments}

All experiments were executed on a workstation equipped with an Intel Core i7-14700F CPU (20 cores / 28 threads) and 32 GB DDR4 RAM. Workloads ran inside Windows Subsystem for Linux 2 (WSL 2) on Windows 11 Pro 22H2, with Ubuntu 24.04.1 LTS. The WSL was constrained to 24 logical processors and 30 GB RAM, with 64 GB swap space. All Circom compilations and Groth16 proof routines used Circom 2.2.2 and SnarkJS v0.7.5 under Node.js v22.15.0. 

We used Groth16 as the underlying zk-SNARK protocol for the key generation, witness generation, proof generation and verification procedure. Poseidon hash function is used for the commitment that binds the inputs of the Data Providers. Our lookup table consists of 128 regular entries and the rest of the entries are set to 0 or $k=\left\lceil\frac{1}{\exp(\varepsilon/2)-1}\right\rceil$ depending on the method \textbf{\textit{set0}} and \textbf{\textit{setk}}. 

Table~\ref{table:Experiment Result e0.5d0}--\ref{table:Experiment Result e1d1} and Figure~\ref{fig:experimental result e0.5d0}--\ref{fig:experimental result e1d1} report the total runtime of our scheme, configured with an exponentiation lookup table to satisfy $(\varepsilon=0.5,\delta=0), (\varepsilon=1,\delta=0), (\varepsilon=0.5,\delta=e^{-31.5}), \text{ and } (\varepsilon=0.5,\delta=e^{-63})$-DP. The measured times include witness generation time \(t_w\), proof generation time \(t_p\), and verification time \(t_v\) for input sizes \(m\) ranging from 1,000 to 7,000. The domain of the median query were fixed to $\mathcal{R}=\{\mathsf{range}_i\}_{i\in[n=100]}=[0,n-1=99]$ in all test cases. Each benchmark was executed 50 times, and the mean values are presented.

\begin{table}[h!]
  \centering
  \begin{tabular}{rrrr}
    \toprule
    \(m\)         & \(t_w\)      & \(t_p\)      & \(t_v\)            \\
    \midrule
    1000               & $0.694$ s        & $17.984$ s       & $16.613$ ms       \\
    2000               & $1.024$ s        & $29.489$ s       & $29.196$ ms       \\
    3000               & $1.446$ s        & $76.886$ s       & $87.684$ ms       \\
    4000               & $2.307$ s        & $117.781$ s      & $166.876$ ms      \\
    5000               & $2.498$ s        & $208.001$ s      & $238.317$ ms      \\
    6000               & $3.149$ s        & $217.078$ s      & $273.572$ ms      \\
    7000               & $3.813$ s        & $259.031$ s      & $277.333$ ms      \\
    \bottomrule
  \end{tabular}
  \caption{Experimental Result $(\varepsilon=0.5, \delta=0)$}
  \label{table:Experiment Result e0.5d0}
\end{table}

\begin{table}[h!]
  \centering
  \begin{tabular}{rrrr}
    \toprule
    \(m\)              & \(t_w\)          & \(t_p\)          & \(t_v\)           \\
    \midrule
    1000               & $0.541$ s        & $14.246$ s       & $19.400$ ms       \\
    2000               & $1.123$ s        & $27.566$ s       & $26.598$ ms       \\
    3000               & $1.708$ s        & $77.733$ s       & $104.073$ ms      \\
    4000               & $2.026$ s        & $121.496$ s      & $146.208$ ms      \\
    5000               & $2.630$ s        & $206.174$ s      & $272.224$ ms      \\
    6000               & $3.179$ s        & $223.024$ s      & $262.876$ ms      \\
    7000               & $4.107$ s        & $263.898$ s      & $389.152$ ms      \\
    \bottomrule
  \end{tabular}
  \caption{Experimental Result $(\varepsilon=1, \delta=0)$}
  \label{table:Experiment Result e1d0}
\end{table}

\begin{table}[h!]
  \centering
  \begin{tabular}{rrrr}
    \toprule
    \(m\)              & \(t_w\)          & \(t_p\)          & \(t_v\)           \\
    \midrule
    1000               & $0.704$ s        & $17.928$ s       & $17.678$ ms       \\
    2000               & $1.030$ s        & $29.631$ s       & $27.286$ ms       \\
    3000               & $1.460$ s        & $77.753$ s       & $100.707$ ms      \\
    4000               & $2.272$ s        & $118.993$ s      & $142.104$ ms      \\
    5000               & $2.441$ s        & $206.990$ s      & $240.001$ ms      \\
    6000               & $3.098$ s        & $219.309$ s      & $264.621$ ms      \\
    7000               & $3.695$ s        & $258.434$ s      & $280.222$ ms      \\
    \bottomrule
  \end{tabular}
  \caption{Experimental Result $(\varepsilon=0.5, \delta=1)$}
  \label{table:Experiment Result e0.5d1}
\end{table}

\begin{table}[h!]
  \centering
  \begin{tabular}{rrrr}
    \toprule
    \(m\)              & \(t_w\)          & \(t_p\)          & \(t_v\)           \\
    \midrule
    1000               & $0.547$ s        & $14.142$ s       & $19.336$ ms       \\
    2000               & $1.064$ s        & $27.570$ s       & $26.258$ ms       \\
    3000               & $1.722$ s        & $78.757$ s       & $92.317$ ms       \\
    4000               & $1.988$ s        & $122.262$ s      & $151.983$ ms      \\
    5000               & $2.572$ s        & $207.425$ s      & $272.152$ ms      \\
    6000               & $3.294$ s        & $220.124$ s      & $240.677$ ms      \\
    7000               & $4.062$ s        & $262.176$ s      & $339.360$ ms      \\
    \bottomrule
  \end{tabular}
  \caption{Experimental Result $(\varepsilon=1, \delta=1)$}
  \label{table:Experiment Result e1d1}
\end{table}

\FloatBarrier

\begin{figure}[h!]
    \centering
    \includegraphics[width=0.65\linewidth]{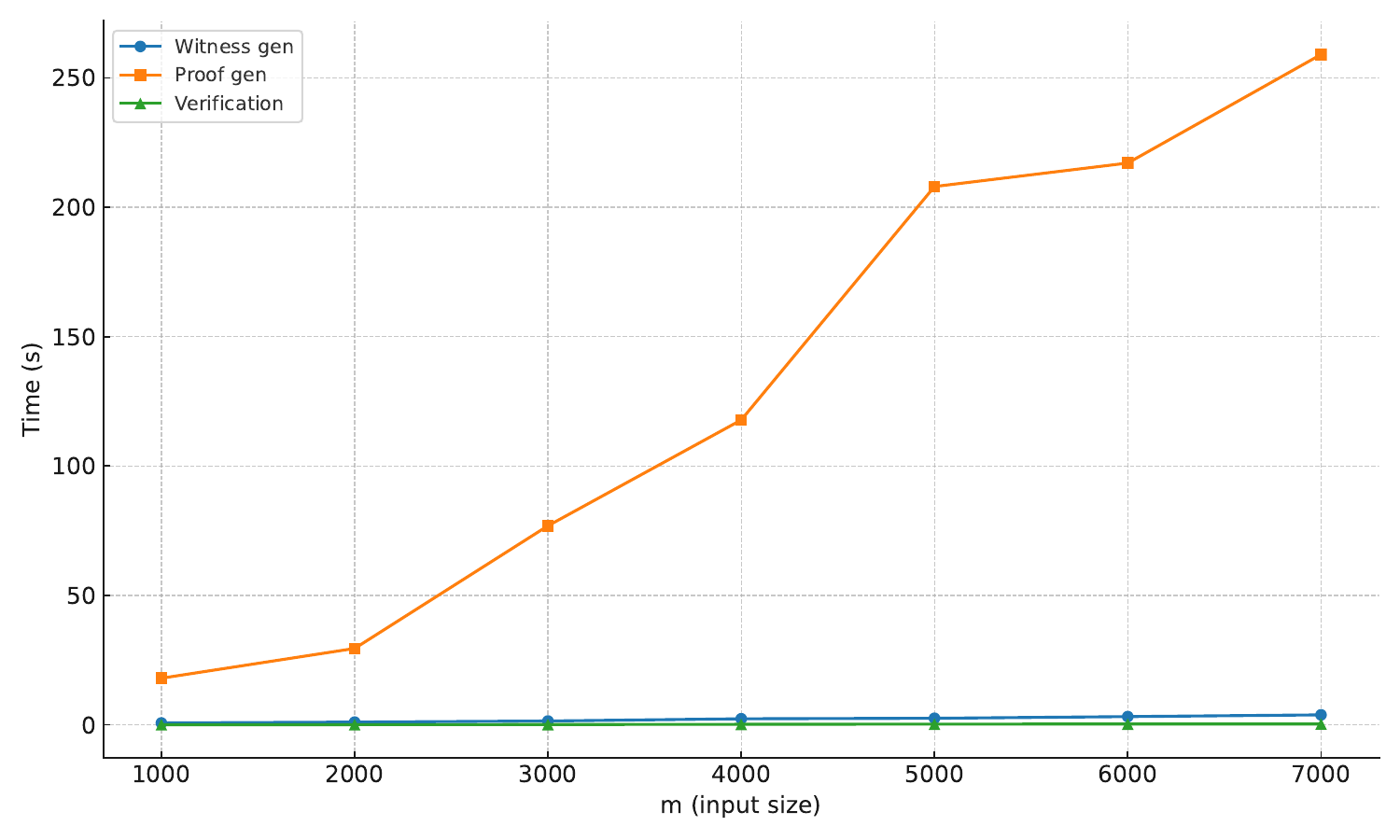}
    \caption{Experimental Result \(\varepsilon=0.5, \delta=0\)}
    \label{fig:experimental result e0.5d0}
\end{figure}

\begin{figure}[h!]
    \centering
    \includegraphics[width=0.65\linewidth]{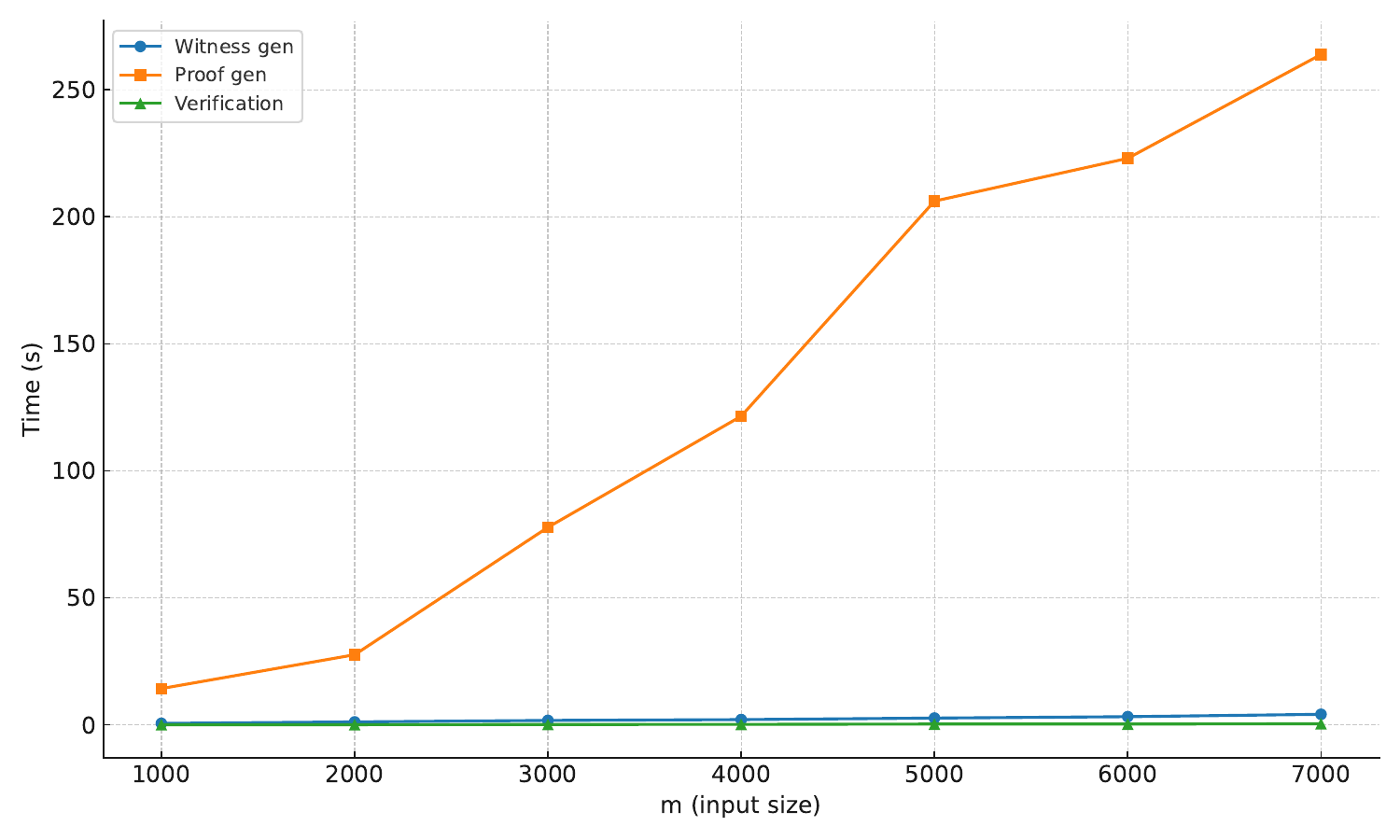}
    \caption{Experimental Result \(\varepsilon=1, \delta=0\)}
    \label{fig:experimental result e1d0}
\end{figure}

\begin{figure}[h!]
    \centering
    \includegraphics[width=0.65\linewidth]{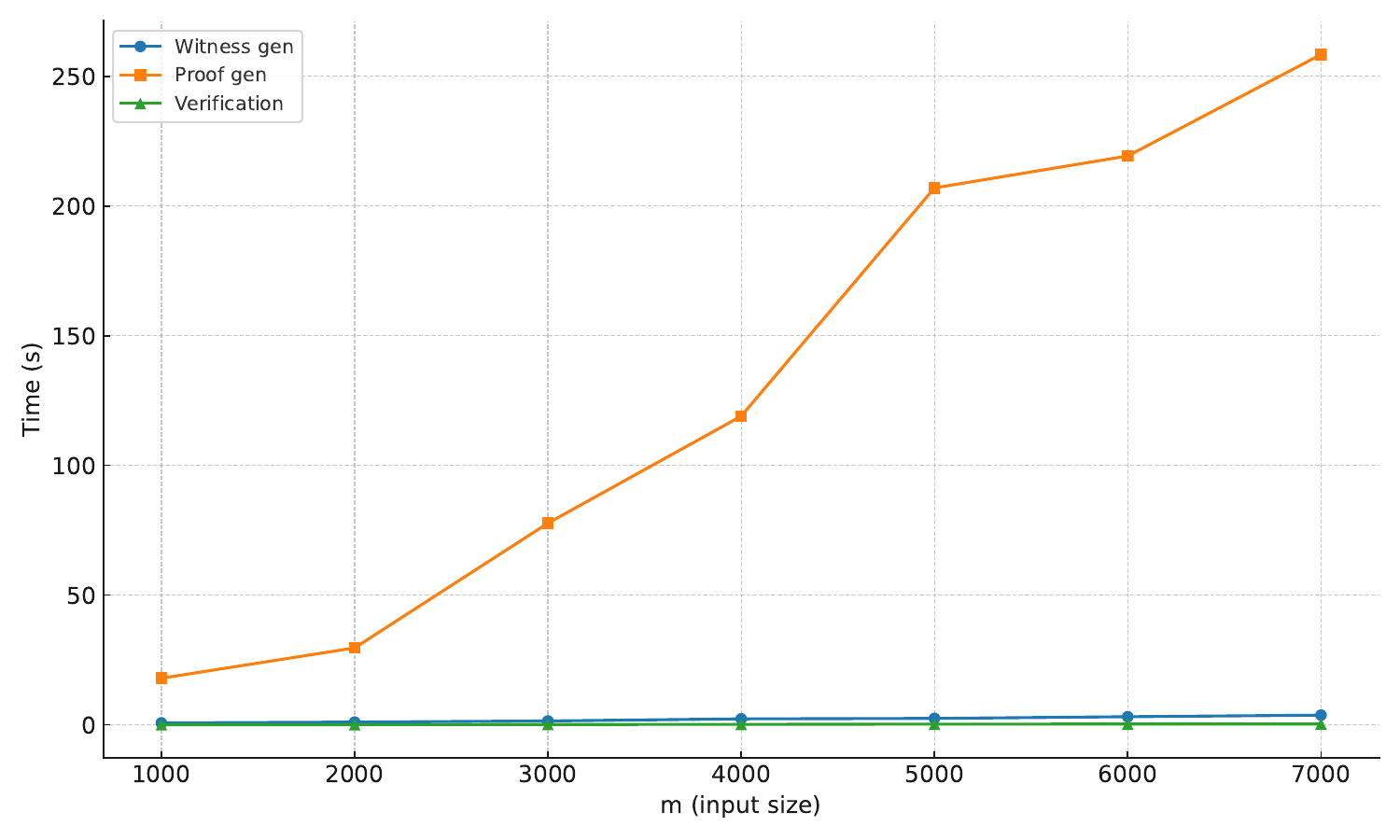}
    \caption{Experimental Result \(\varepsilon=0.5, \delta= e^{-31.5}\)}
    \label{fig:experimental result e0.5d1}
\end{figure}

\begin{figure}[h!]
    \centering
    \includegraphics[width=0.65\linewidth]{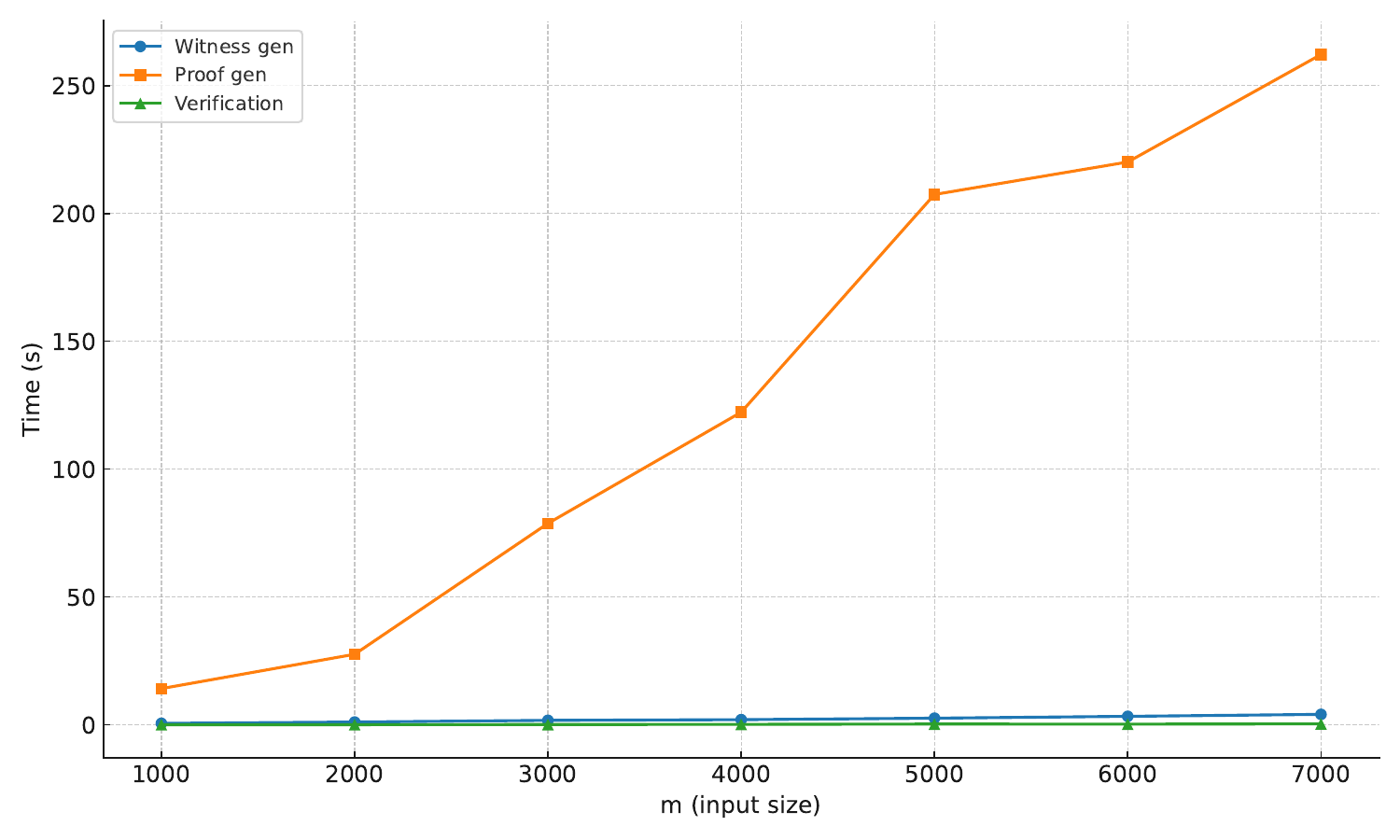}
    \caption{Experimental Result \(\varepsilon=1, \delta=e^{-63}\)}
    \label{fig:experimental result e1d1}
\end{figure}

\section{Conclusion}

We present \(\mathsf{VerExp}\), the first verifiable exponential mechanism for median estimation. First we formalize the setting by suggesting a security model with four parties: Data Providers, Data Analyst, Verifier and Public Bulletin Board. Then we define a pipeline of data flow with interaction among the parties, and design an arithmetic circuit for exponential mechanism with a scaled inverse CDF sampling procedure. Our construction guarantees strong privacy and security with high utility. We implement our scheme using Groth16 as our zkSNARK of choice and achieve public verifiability and very short verification time. 

Nevertheless, this work is tailored for median estimation due to the dependency on discrete rank-based utility. Extending for real-valued arbitrary utility function remains an open problem. Additionally, reducing the circuit complexity of \(O(m\cdot n)\) is another key direction for future research. 

\section*{Acknowledgments}
This work was supported by the New Faculty Startup Fund from Seoul National University.

\bibliographystyle{ACM-Reference-Format}
\bibliography{reference}

\clearpage
\appendix

\section{Pseudocode for Modules}\label{section:pseudocodes}

\begin{table}[h!]
  \centering
  \begin{tabular}{p{0.95\linewidth}}
    \toprule
        \(\mathsf{Bind}(\{\mathsf{input}_i\}_{i\in[m]}, \{\mathsf{rand}_i\}_{i\in[m]})\) \\ 
    \midrule
    \begin{minipage}{\linewidth}
      \begin{algorithmic}[1]
        \FOR{\(i\in[m]\)}
            \STATE \(\mathsf{com}_i\leftarrow\mathsf{hash}(\mathsf{input}_i, \mathsf{rand}_i)\)
        \ENDFOR
        \STATE \textbf{output} \(\{\mathsf{com}_i\}_{i\in[m]}\)
      \end{algorithmic}
    \end{minipage}
    \\ \bottomrule 
  \end{tabular}
  \caption{\(\mathsf{Bind}\) Module}
  \label{table:Input Bind} 
\end{table}

\begin{table}[h!]
  \centering
  \begin{tabular}{p{0.95\linewidth}}
    \toprule
        \(\mathsf{Util}(\{\mathsf{range}_i\}_{i\in[n]}, \{\mathsf{input}_i\}_{i\in[m]})\) \\ 
    \midrule
    \begin{minipage}{\linewidth}
      \begin{algorithmic}[1]
        \FOR{\(i\in[n]\)}
            \FOR{\(j\in[m]\)}
                \IF{\(\mathsf{range}_i>\mathsf{input_j}\)}
                    \STATE \(ind_{i,j}\leftarrow 1\)
                \ELSE
                    \STATE \(ind_{i,j}\leftarrow 0\)
                \ENDIF
            \ENDFOR
            \STATE \(\mathsf{util}_i=\left|\frac{n-1}{2}-\sum_{j=0}^{m-1}ind_{i,j}\right|\)
        \ENDFOR
        \STATE \(\{\mathsf{util}'_i\}_{i\in[n]}\leftarrow \mathsf{SubMin}(\{\mathsf{util}_i\}_{i\in[n]})\)
        \STATE \textbf{output} \(\{\mathsf{util}'_i\}_{i\in[n]}\)
      \end{algorithmic}
    \end{minipage}
    \\ \bottomrule
  \end{tabular}
  \caption{\(\mathsf{Util}\) Module}
  \label{table:Utility} 
\end{table}

\begin{table}[h!]
  \centering
  \begin{tabular}{p{0.95\linewidth}}
    \toprule
        \(\mathsf{SubMin}(\{\mathsf{util}_i\}_{i\in[n]})\) \\ 
    \midrule
    \begin{minipage}{\linewidth}
      \begin{algorithmic}[1]
        \STATE \(min \leftarrow \mathsf{util}_0\)
        \FOR{\(1\le i \le n-1\)}
            \IF{\(min > \mathsf{util}_{i}\)}
                \STATE \(min\leftarrow \mathsf{util}_i\)
            \ENDIF
        \ENDFOR

        \FOR{\(i\in[n]\)}
            \STATE \(\mathsf{util}'_i\leftarrow \mathsf{util}_i-min\)
        \ENDFOR
        \STATE \textbf{output} \(\{\mathsf{util}'_i\}_{i\in[n]}\)
      \end{algorithmic}
    \end{minipage}
    \\ \bottomrule
  \end{tabular}
  \caption{\(\mathsf{SubMin}\) Module}
  \label{table:SubMin} 
\end{table}

\begin{table}[h!]
  \centering
  \begin{tabular}{p{0.95\linewidth}}
    \toprule
        \(\mathsf{ExpLookup}(\mathsf{util}'_i)\) \\ 
    \midrule
    \begin{minipage}{\linewidth}
      \begin{algorithmic}[1]
        \IF{\(\mathsf{util}'_i>l-1\)}
            \STATE \(\mathsf{expval}_i\leftarrow k=\left\lceil\frac{1}{\exp(\varepsilon/2)-1}\right\rceil\)
        \ELSE
            \STATE \(\mathsf{expval}_i\leftarrow T[\mathsf{util}'_i]\)
        \ENDIF
        \STATE \textbf{output} \(\mathsf{expval}_i\)
      \end{algorithmic}
    \end{minipage}
    \\ \bottomrule
  \end{tabular}
  \caption{\(\mathsf{ExpLookup}\) Module (\textit{\textbf{setk}})}
  \label{table:ExpLookup setk} 
\end{table}

\begin{table}[h!]
  \centering
  \begin{tabular}{p{0.95\linewidth}}
    \toprule
        \(\mathsf{ExpLookup}(\mathsf{util}'_i)\) \\ 
    \midrule
    \begin{minipage}{\linewidth}
      \begin{algorithmic}[1]
        \IF{\(\mathsf{util}'_i>l-1\)}
            \STATE \(\mathsf{expval}_i\leftarrow 0\)
        \ELSE
            \STATE \(\mathsf{expval}_i\leftarrow T[\mathsf{util}'_i]\)
        \ENDIF
        \STATE \textbf{output} \(\mathsf{expval}_i\)
      \end{algorithmic}
    \end{minipage}
    \\ \bottomrule
  \end{tabular}
  \caption{\(\mathsf{ExpLookup}\) Module (\textit{\textbf{set0}})}
  \label{table:ExpLookup set0} 
\end{table}

\begin{table}[h!]
  \centering
  \begin{tabular}{p{0.95\linewidth}}
    \toprule
        \(\mathsf{InverseCDF}(\{s_i\}_{i\in[n]}, \{\mathsf{range}_i\}_{i\in[n]}, \rho)\) \\ 
    \midrule
    \begin{minipage}{\linewidth}
      \begin{algorithmic}[1]
        \FOR{\(i\in[n]\)}
            \IF{\(s_i>\rho\)}
                \STATE \(sig_i \leftarrow 1\)
            \ELSE
                \STATE \(sig_i \leftarrow 0\)
            \ENDIF
        \ENDFOR
        \STATE \(sig'_0\leftarrow sig_0\)
        \FOR{\(1\le i\le n-1\)}
            \STATE \(sig'_i\leftarrow sig_{i-1} + sig_i \mod{2}\)
        \ENDFOR
        \FOR{\(i\in[n]\)}
            \STATE \(sig_i''=\mathsf{range}_i\cdot sig_i'\)
        \ENDFOR
        \STATE \(\mathsf{med}\leftarrow \sum_{i=0}^{n-1} sig''_i\)
        \STATE \textbf{output} \(\mathsf{med}\)
      \end{algorithmic}
    \end{minipage}
    \\ \bottomrule
  \end{tabular}
  \caption{\(\mathsf{InverseCDF}\) Module}
  \label{table:InverseCDF} 
\end{table}

\FloatBarrier

\clearpage

\section{Figures for The Pipeline And The Circuits} \label{section:figures}

\begin{figure}[h!]
    \centering
    \includegraphics[width=0.85\linewidth]{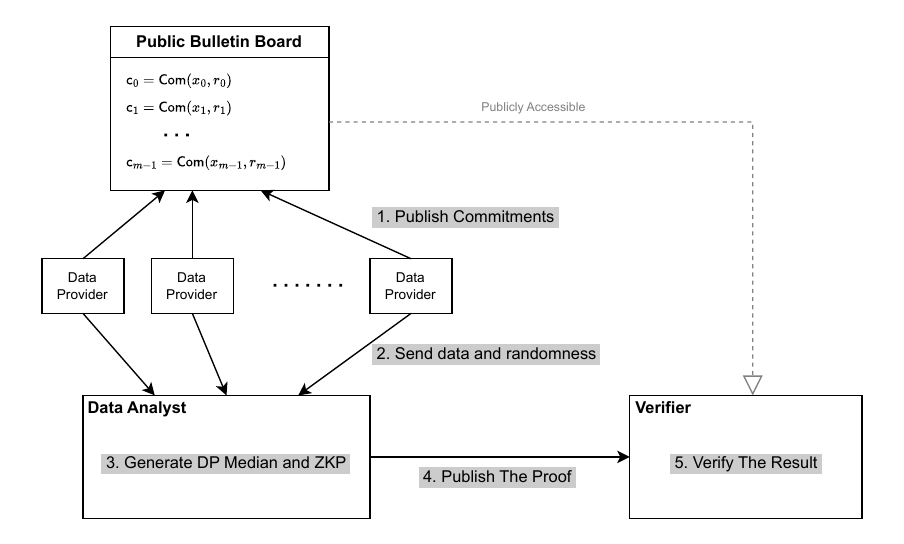}
    \caption{Pipeline}
    \label{fig:Pipeline}
\end{figure}

\vspace{-1em}

\begin{figure}[h!]
    \centering
    \includegraphics[width=0.7\linewidth]{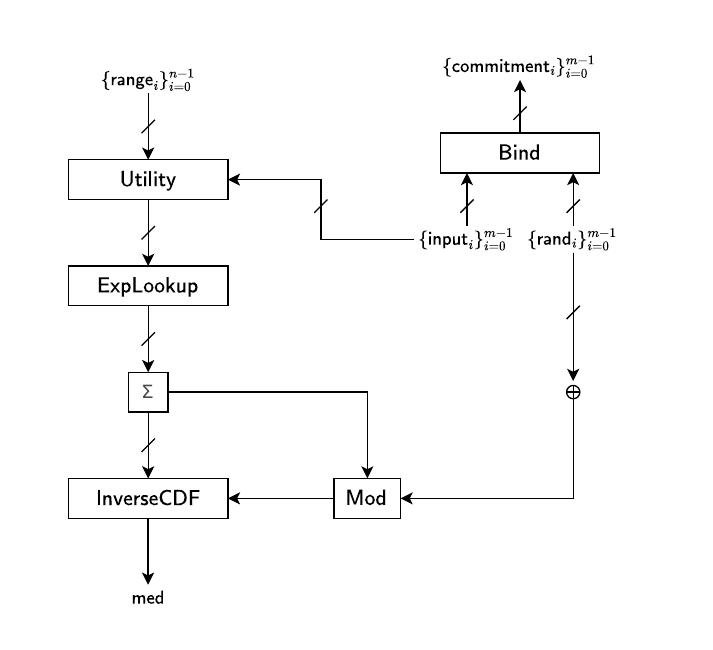}
    \caption{The Main Circuit \(\mathcal{C}\)}
    \label{fig:Overall circuit}
\end{figure}

\vspace{-1em}

\FloatBarrier

\begin{figure}[h!]
    \centering
    \includegraphics[width=0.9\linewidth]{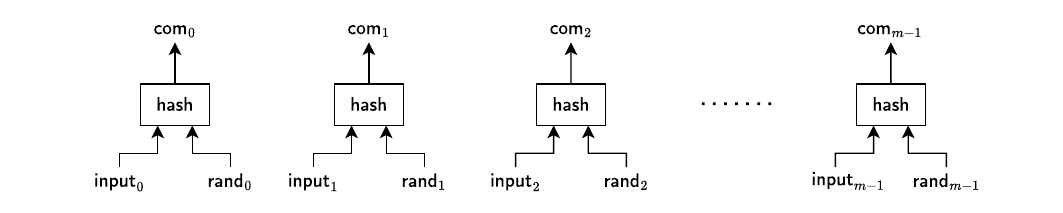}
    \caption{\(\mathsf{Bind}\) Module}
    \label{fig:Binder}
\end{figure}

\vspace{-1em}

\begin{figure}[h!]
    \centering
    \includegraphics[width=0.9\linewidth]{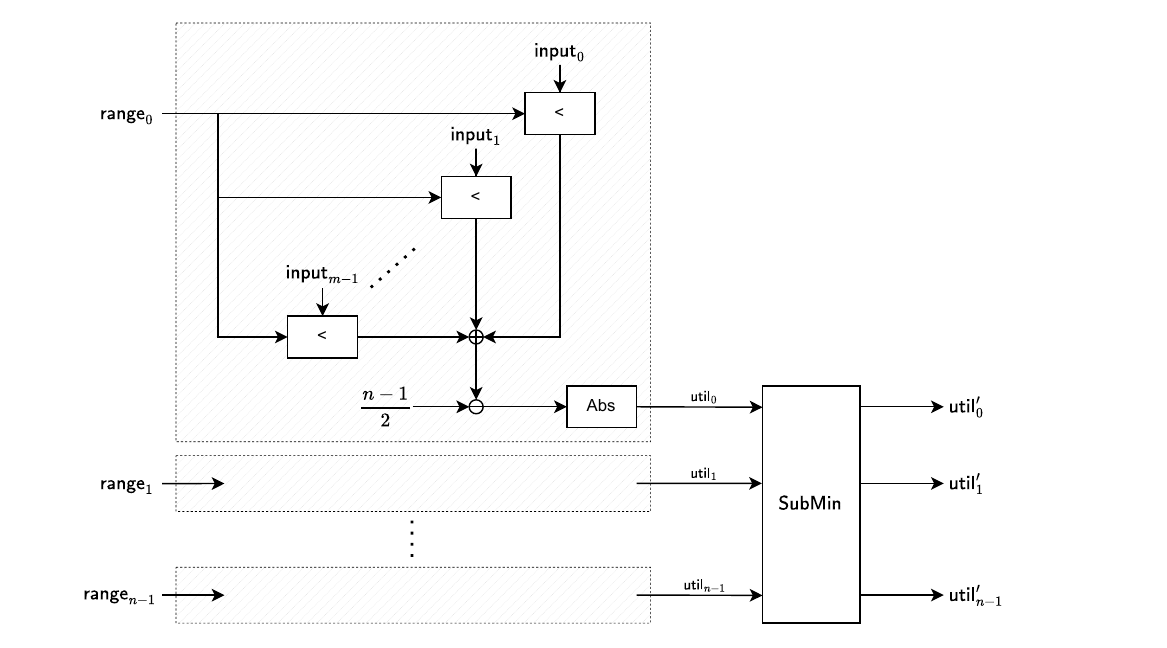}
    \caption{\(\mathsf{Util}\) Module}
    \label{fig:Util}
\end{figure}

\begin{figure}[h!]
    \centering
    \includegraphics[width=0.75\linewidth]{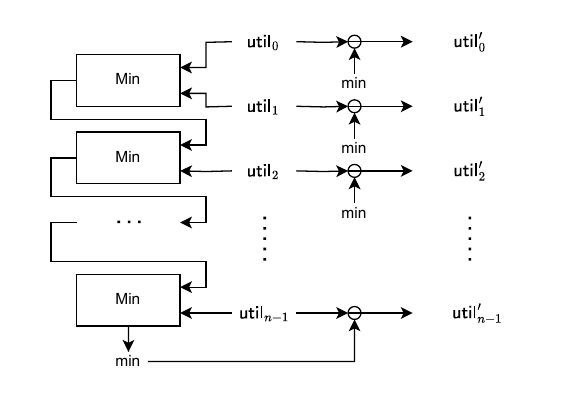}
    \caption{\(\mathsf{SubMin}\) Module}
    \label{fig:SubMin}
\end{figure}

\begin{figure}[h!]
    \centering
    \includegraphics[width=0.95\linewidth]{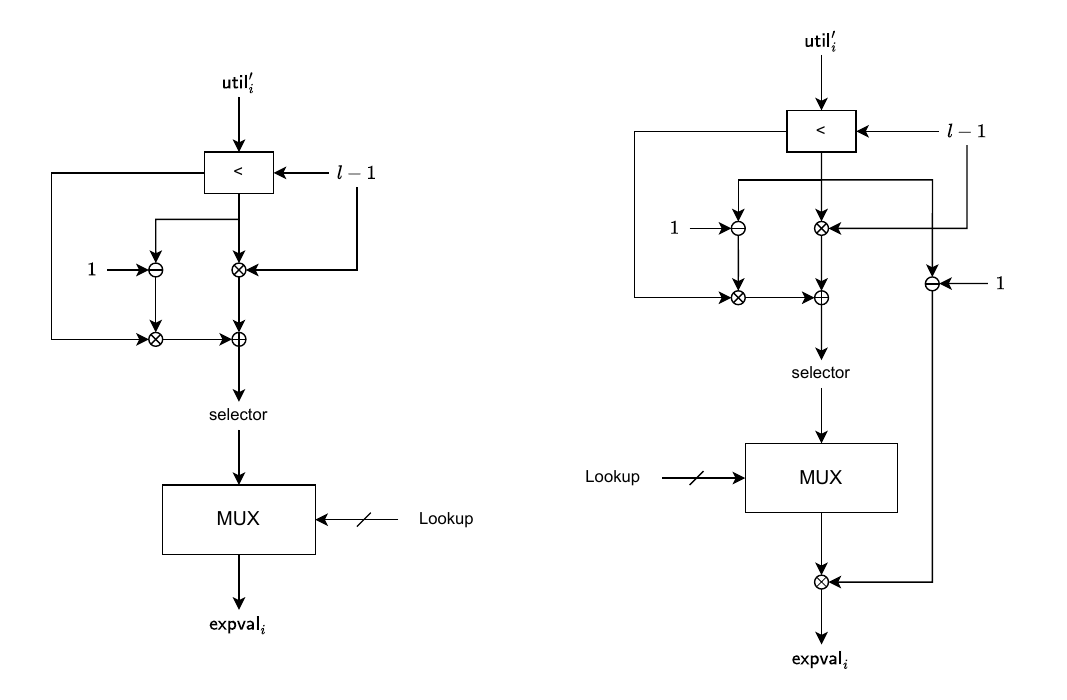}
    \caption{\(\mathsf{ExpLookup}\) Module}
    \label{fig:ExpLookup}
\end{figure}

\begin{figure}[h!]
    \centering
    \includegraphics[width=1\linewidth]{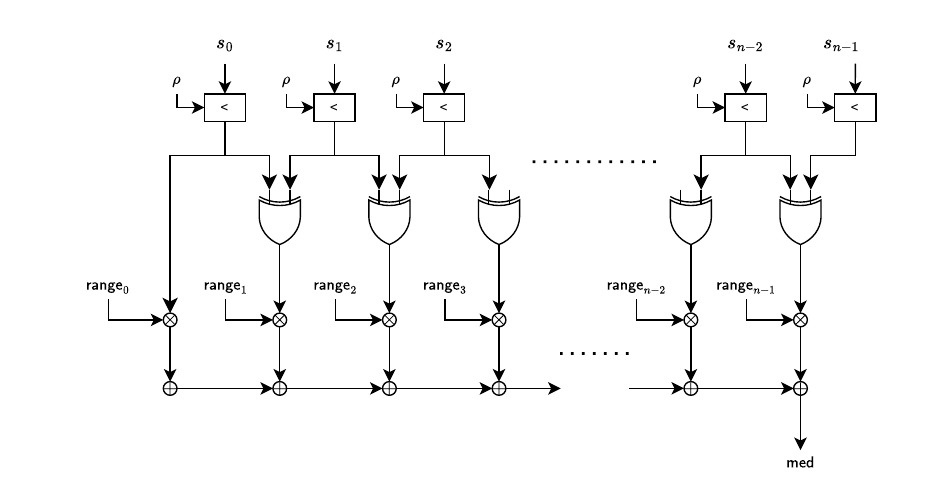}
    \caption{\(\mathsf{InverseCDF}\) Module}
    \label{fig:InverseCDF}
\end{figure}

\FloatBarrier

\clearpage

\section{Lemmas}

\begin{lemma} \label{lem: approx err}
    Let $M$ be a sufficiently large integer and $a>1$ be a constant. Let Sequences $A[i]$ and $B[i]$ be defined as follows.
    \begin{align*}
        A[i] = 
        \begin{cases}
          M & \text{if } i=0 \\
          \left\lceil \frac{A[i-1]}{a} \right\rceil  & \text{otherwise}
        \end{cases} ,\;\;\;\;\;\;\;\;\;\;\;\;
        B[i] = N\cdot e^{-i} \notag
    \end{align*}
    Then, for all $i$, $\left|A[i]-B[i]\right|<\frac{a}{a-1}$
\end{lemma}

\begin{proof}
    We can express each element in the sequence $A$ as follows.
    \begin{align*}
        A[i] = \frac{A[i-1]}{a}+\alpha_i \;(0 \le \alpha_i < 1)
    \end{align*}
    Let $E[i]$ be the error of $i$-th term between $A$ and $B$. Then,
    \begin{align*}
        E[i] &= A[i]-B[i] \\
        &= \left(\frac{T[i-1]}{a}+\alpha_i\right)-\frac{B[i-1]}{a} \\
        &= \frac{E[i-1]}{a} + \alpha_i
    \end{align*}
    By solving the recurrence relation, we get 
    \begin{align*}
        E[i]=\sum_{j=1}^{i}\frac{\alpha_i}{e^{i-j}}
    \end{align*}
    Since $0\le\alpha_i<1$ for all $i$, we get
    \begin{align*}
        E[i] \le \sum_{j=1}^{i}\frac{1}{a^{i-j}} = \frac{a}{a-1}
    \end{align*}
\end{proof}

\FloatBarrier

\begin{lemma} \label{lem:randomness rho is comp}
    Let the statistical distribution of \(\rho\) be \(\mathcal{U}_{\mathsf{comp}}(\mathbb{Z}_{s_{n-1}})\).
    If \(s_{n-1} \ll p\),
    \[
        \mathcal{U}_{\mathsf{comp}}(\mathbb{Z}_{s_{n-1}}) \overset{c}{\approx} \mathcal{U}(\mathbb{Z}_{s_{n-1}})
    \]
\end{lemma}

\begin{proof}
    Recall that we assume that each \(\mathcal{D}_i\) generates \(r_i\) from its private randomness, uniformly at random.
    We obtain \(\rho\) as follows:
    \[
        \rho = \left(\sum_{i=0}^{m-1}r_i\mod{p}\right) \mod{s_{n-1}}
    \]
    Let \(z\equiv p \mod{s_{n-1}}\) such that \(0\le z < s_{n-1}\).
    Then 
    \[\Pr[\rho=x]=
        \begin{cases}
          \frac{1}{p}+\frac{p-z}{ps_{n-1}}      & \text{if } 0 \le x \le z,\\
          \frac{p-z}{ps_{n-1}}       & \text{otherwise } x < 0.
        \end{cases}
    \]
    By the security of the underlying zkSNARK scheme, \(p(\lambda)\) is an exponentially large prime. Therefore by the definition of statistical distance \(\Delta\), 
    \begin{align*}
        \Delta(\mathcal{U}_{\mathsf{comp}}(\mathbb{Z}&_{s_{n-1}}), \mathcal{U}(\mathbb{Z}_{s_{n-1}}))\\ 
        &= \frac{1}{2}\sum_{x\in\mathbb{Z}_{s_{n-1}}} \left| \Pr[X=x]-\Pr[Y=x]\right|\\
        &= \frac{1}{2}\left(\sum_{x=0}^{z-1}\frac{z-s_{n-1}}{ps_{n-1}}+\sum_{x=z}^{s_{n-1}-1}\frac{z}{ps_{n-1}}\right) \\
        &= \frac{z(s_{n-1}-z)}{ps_{n-1}} \\
        &\le \frac{s_{n-1}}{4p}
    \end{align*}
    when \(X\sim \mathcal{U}_{\mathsf{comp}}(\mathbb{Z}_{s_{n-1}})\) and \(Y\sim \mathcal{U}(\mathbb{Z}_{s_{n-1}})\).
    
    Therefore if \(s_{n-1} \ll p\), \(\Delta(\mathcal{U}_{comp}(\mathbb{Z}_{s_{n-1}}), \mathcal{U}(\mathbb{Z}_{s_{n-1}}))\) is negligible. Furthermore, since statistical negligible similarity implies computational indistinguishability, we can conclude that \(\mathcal{U}_{comp}(\mathbb{Z}_{s_{n-1}})\) is computationally indistinguishable from the uniform distribution, \(\mathcal{U}(\mathbb{Z}_{s_{n-1}})\).
\end{proof}

\FloatBarrier

\begin{lemma} \label{lem:randomness change}
    Assume a mechanism \(\mathcal{M}\) satisfies \((\varepsilon,\delta)\)-DP when the underlying randomness distribution is \(\mathcal{U}(S)\) for a set \(S\). The mechanism \(\mathcal{M}\) satisfies \((\varepsilon,\delta)\)-computational DP when the underlying randomness distribution is \(\mathcal{U}_{\mathsf{comp}}(S)\) such that 
    \[
        \mathcal{U}_{\mathsf{comp}}(S)\overset{c}{\approx}\mathcal{U}(S)
    \]
\end{lemma}

\begin{proof}
    Let \(\mathcal{DB}\) and \(\mathcal{DB}'\) be adjacent databaes. Let \(\mathcal{M}_\mathsf{uni}\) be the mechanism \(\mathcal{M}\) with randomness \(\mathcal{U}(S)\) and \(\mathcal{M}_\mathsf{comp}\) be the mechanism \(\mathcal{M}\) with randomness \(\mathcal{U}_\mathsf{comp}(S)\).
    
    For any PPT distinguisher \(A\),
    \begin{align*}
        &\left|\Pr[A(\mathcal{M}_\mathsf{comp}(\mathcal{DB})=1)-\Pr[A(\mathcal{M}_{\mathsf{comp}}(\mathcal{DB}')=1)]] \right| \\
        &\le \left| \Pr[A(\mathcal{M}_\mathsf{comp}(\mathcal{DB})=1)-\Pr[A(\mathcal{M}_{\mathsf{uni}}(\mathcal{DB})=1)]] \right| \\
        &\;\;\;+ \left| \Pr[A(\mathcal{M}_\mathsf{uni}(\mathcal{DB})=1)-\Pr[A(\mathcal{M}_{\mathsf{uni}}(\mathcal{DB}')=1)]] \right| \\
        &\;\;\;+ \left| \Pr[A(\mathcal{M}_\mathsf{uni}(\mathcal{DB}')=1)-\Pr[A(\mathcal{M}_{\mathsf{comp}}(\mathcal{DB}')=1)]] \right| \\
        &\le e^\varepsilon+\delta+\mathsf{negl}(\lambda)
    \end{align*}

    By the definition of computational differential privacy, \(\mathcal{M}_\mathsf{comp}\) is \((\varepsilon,\delta)\)-computational DP. If we assume that the statistical distance of the randomnesses used for \(\mathcal{M}_\mathsf{uni}\) and \(\mathcal{M}_\mathsf{comp}\) is \(\gamma\) and don't omit the negligible term, we can say that \(\mathcal{M}_\mathsf{comp}\) is \((\varepsilon,\delta+2\gamma)\)-computational DP.
\end{proof}

\FloatBarrier

\begin{lemma} \label{lem:randomness switch bad prob}
    Let there be two statistical distribution \(R,R'\) whose sample space is the same. Assume that the statistical distance \(\Delta\) is:
    \[
        \Delta(R,R')=\gamma
    \]
    Let \(\mathcal{M}\) be a mechanism that uses \(R\) or \(R'\) as the source of its randomness. We denote each instance as \(\mathcal{M}_R\) and \(\mathcal{M}_{R'}\).
    Let \(E\) be the event that is defined solely in terms of the mechanism's output. Then,
    \[
        |\Pr[E\leftarrow\mathcal{M}_R]-\Pr[E\leftarrow \mathcal{M}_{R'}]| \le \gamma
    \]
\end{lemma}

\begin{proof}
    By the definition of statistical distance, 
    \begin{align*}
        \Delta(R,R') &= \frac{1}{2}\sum_x|\Pr[r=x]-\Pr[r'=x]| \\
        &=  \max_E|\Pr[E\leftarrow\mathcal{M}_R]-\Pr[E\leftarrow \mathcal{M}_{R'}]|
    \end{align*}
\end{proof}

\FloatBarrier


\end{document}